\documentclass[11pt,a4paper,reqno]{amsart}
\usepackage[hmargin={2.7cm,2.7cm},vmargin={2.5cm,2.5cm},centering]{geometry}
\usepackage{amssymb,amsmath,amsthm,amsfonts,amscd}
\usepackage{graphicx,xcolor}
\usepackage{mathrsfs}
\usepackage[unicode,psdextra]{hyperref}
\hypersetup{pdfborder={0 0 0.06},citebordercolor=[rgb]{0.8196,0.2275,0.5098},linkbordercolor=[rgb]{0.1765,0.5490,0.8118},urlbordercolor=[rgb]{0.7059,0.5333,0.1137}}
\usepackage{tikz-cd}
\usepackage[utf8]{inputenc}
\usepackage{calc}
\usepackage[english]{babel}
\usepackage{dsfont}
\usepackage[longnamesfirst,numbers,sort&compress]{natbib}
\usepackage{enumitem}
\usepackage{setspace}
\usepackage[font={scriptsize,bf}]{caption}
\usepackage{changepage}
\usepackage{booktabs}
\usepackage{comment}
\usepackage[sort&compress,capitalise,nameinlink]{cleveref}
\crefname{section}{\textsection}{\textsection}
\crefname{subsection}{\textsection}{\textsection}
\crefname{subsubsection}{\textsection}{\textsection}
\crefname{paragraph}{\textparagraph}{\textparagraph}
\crefname{thm}{Theorem}{Theorem}
\crefname{hyp}{Assumption}{Assumption}
\usepackage[cal=euler,scr=boondoxo,bb=ams,calscaled=1.07,scrscaled=1.07]{mathalfa}
\makeatletter
\def\mathalfa@frakscaled{s*[1]}
\makeatother
\DeclareFontFamily{U}{euf}{}%
\DeclareFontShape{U}{euf}{m}{n}{<-7>\mathalfa@frakscaled eufm5
  <7-9>\mathalfa@frakscaled eufm7
  <9->\mathalfa@frakscaled eufm10}{}%
\DeclareFontShape{U}{euf}{b}{n}{<-7>\mathalfa@frakscaled eufb5
  <7-9>\mathalfa@frakscaled eufb7
  <9->\mathalfa@frakscaled eufb10}{}%
\DeclareMathAlphabet{\mathfrak}{U}{euf}{m}{n}

\DeclareMathOperator{\ran}{Ran}

\DeclareMathOperator*{\slim}{s-lim}
\DeclareMathOperator*{\wlim}{w-lim}
\renewcommand{\Im}{\mathrm{Im}}
\renewcommand{\Re}{\mathrm{Re}}
          \newtheorem{thm}{Theorem}[section]
          \newtheorem{proposition}[thm]{Proposition}
          \newtheorem{lemma}[thm]{Lemma}
          \newtheorem{corollary}[thm]{Corollary}
          \newtheorem{definition}[thm]{Definition}
          
          \theoremstyle{definition}

          \newtheorem{remark}[thm]{Remark}
          \newtheorem{hyp}{Assumption}

\renewcommand{\setminus}{\smallsetminus}
\usepackage{csquotes}
\setlist[enumerate,3]{itemsep=4mm, label=\Alph*.,ref=\Alph*}
\setlist[enumerate,2]{itemsep=3mm, label=\alph*., topsep=3mm, ref=\alph*}
\setlist[enumerate,1]{itemsep=2mm, label=\roman*), topsep=2mm, ref=\roman*)}
\setlist[itemize,1]{itemsep=4mm}
\setlist[itemize,2]{itemsep=3mm, label=$*$, topsep=3mm}
\setlist[itemize,3]{itemsep=2mm, label=$\diamond$, topsep=2mm}
\newcommand{\RR}{\mathbb{R}}

\newcommand{\NN}{\mathbb{N}}

\newcommand{\der}{\partial}

\newcommand{\R}{\mathbb{R}}

\newcommand{\N}{\mathbb{N}}

\newcommand{\F}{\mathcal{F}}

\begin{document}
\onehalfspacing{} \bibliographystyle{natalpha}

\newenvironment{sistema}%
{\left\lbrace\begin{array}{@{}l@{}}}%
    {\end{array}\right.}

\title[Semiclassical analysis of asymptotic fields in the Yukawa theory]{
  Semiclassical analysis of quantum asymptotic fields in the Yukawa theory}

\author[Z.\ Ammari]{Zied Ammari}

\address{Univ Rennes, [UR1], CNRS, IRMAR - UMR 6625 \\F-35000 Rennes, France}

\email{zied.ammari@univ-rennes1.fr}

\urladdr{}

\author[M.\ Falconi]{Marco Falconi}

\address{Politecnico di Milano \\ D-Mat\\ P.zza Leonardo da Vinci \\ 20133
  Milano \\ Italy}

\email{marco.falconi@polimi.it}

\urladdr{}

\author[M.\ Olivieri]{Marco Olivieri}

\address{Fakult\"{a}t f\"{u}r Mathematik\\ Karlsruher Institut f\"{u}r
  Technologie\\ D-76128, Karlsruhe\\ Germany}

\email{marco.olivieri@kit.edu}

\urladdr{}



\date{\today}

\begin{abstract}
  In this article, we study the asymptotic fields of the Yukawa
  particle-field model of quantum physics, in the semiclassical regime $\hslash\to
  0$, with an interaction subject to an ultraviolet cutoff. We show that the
  transition amplitudes between final (respectively initial) states converge
  towards explicit quantities involving the outgoing (respectively incoming)
  wave operators of the nonlinear Schrödinger--Klein--Gordon (S-KG)
  equation. Thus, we rigorously link the scattering theory of the Yukawa
  model to that of the Schrödinger--Klein--Gordon equation. Moreover, we prove
  that the asymptotic vacuum states of the Yukawa model have a phase space
  concentration property around classical radiationless solutions. Under
  further assumptions, we show that the S-KG energy admits a unique minimizer
  modulo symmetries and identify exactly the semiclassical measure of Yukawa
  ground states.  Some additional consequences of asymptotic completeness are
  also discussed, and some further open questions are raised.
\end{abstract}

\maketitle

% \tableofcontents

\section{Introduction}

The fundamental interaction between matter and the electromagnetic field
forms the underlying theoretical basis of various branches of physics and
engineering (quantum optics, atomic spectroscopy, magnetized plasma,
spacecraft semi-conductors, \dots). Such interaction is modeled by a variety
of reduced mathematical models. For instance, in classical mechanics the
light-matter interaction is simply described via Newton's law and the Lorentz
force, while in quantum mechanics it is described by quantum electrodynamics
(QED).  An alternative approach for studying light-matter interactions,
depending on the scale of the system, is provided by semiclassical
electrodynamics, where quantum matter interacts with classical
electromagnetic fields, according for example to the Maxwell--Schr\"odinger or
Maxwell–Bloch equations.  Similarly in high energy physics, the quantum and
semiclassical descriptions provide two different models of the strong nuclear
force via the Yukawa theory, depending on whether the meson field is treated
as a quantum or classical object.  Although there is consensus about the fact
that the quantum and the semiclassical descriptions of matter-field
interactions are intimately related by Bohr's correspondence principle, there
are few rigorous derivations of such relationship \citep[see,
\emph{e.g.},][and references therein contained]{ammari2014jsp,
  ammari2017sima, Leopold:2018aa, MR4159563, carlone2021sima,
  correggi2019jst, correggi2020arxiv, correggi2021jems}.  From both a
conceptual and a practical standpoint, it is interesting to question the
consistency of these different models of particle-field interactions and to
introduce an asymptotic analysis for \emph{transition} and \emph{scattering
  amplitudes} in terms of classical quantities. Recall that the scattering
amplitudes correspond to the main measurable quantities in experimental
physics, and their study is quite challenging particularly for realistic
models.

According to the standard formalism of quantum field theory, the
time-asymptotics of quantum dynamics is encoded in the asymptotic fields and
in the $S$-matrix \cite{MR142335}. In particular, the transition amplitudes
for various scattering processes are described by the following expectations
values
\begin{equation}
  \label{eq.int.1}
  \langle f^\alpha_\hslash, f^\beta_\hslash\rangle  \quad \text{ and } \quad \langle i^\alpha_\hslash, i^\beta_\hslash\rangle\,,
\end{equation}
where $|f^\alpha_\hslash\rangle$ are final (out) states and $|i^\beta_\hslash\rangle$ are initial (in) states
defined through the asymptotic fields as,
\begin{equation}
  \label{eq.int.ts}
  |f^\alpha_\hslash\rangle=\prod_{j} a^{+,*}_\hslash(\alpha_j) |\Omega_\hslash\rangle, \qquad \text{ and } \qquad |i^\beta_\hslash\rangle=\prod_{k} a^{-,*}_\hslash(\beta_k) |\Omega_\hslash\rangle\,,
\end{equation}
where $|\Omega_\hslash\rangle$ is a prepared vector state, $\alpha_j,\beta_k$ are smooth functions and
$a_{\hslash}^{\pm,*},a_{\hslash}^{\pm}$ are the asymptotic creation-annihilation operators
respectively.  On the other hand, the scattering amplitudes are defined
generally as the quantities
\begin{equation}
  \label{eq.int.amp}
  \quad \langle f^\alpha_\hslash, i^\beta_\hslash\rangle\,.
\end{equation}
Specifically, if one takes $|\Omega_\hslash\rangle$ to be an asymptotic vacuum (ground) state
of the interacting system, the transition amplitudes \eqref{eq.int.1} can be
computed exactly thanks to the canonical commutation relations
$$
a^{\pm}_\hslash(\alpha_j) a^{\pm,*}_\hslash(\beta_k)- a^{\pm,*}_\hslash(\beta_k)a^{\pm}_\hslash(\alpha_j)= \hslash \langle \alpha_j, \beta_k\rangle \mathds{1}\,,
$$
and Wick's theorem. However, for interacting quantum field theories it is not
possible in general to give exact closed formulas for these transition and
scattering amplitudes.  Therefore, one resorts to perturbation theory and
semiclassical analysis in order to provide expansions of these amplitudes
respectively in terms of the coupling constant or the semiclassical
parameter.  One of the remarkable results of QFT is the
\emph{Lehmann-Symanzik-Zimmermann} (LSZ) reduction formula which relates the
scattering amplitudes \eqref{eq.int.amp} to time-ordered correlation
functions, thus formally allowing to derive a perturbative expansion by means
of Feynman diagrams. An important question is then to prove rigorously an
expansion formula for the transition and scattering amplitudes
\eqref{eq.int.1}-\eqref{eq.int.amp} with respect to the semiclassical
parameter $\hslash$.  It is worth noting that in principle the semiclassical limit
$\hslash\to 0$ in \eqref{eq.int.amp} yields the tree diagrams of the $S$-matrix
elements while the loop diagrams provide higher order corrections in terms of
$\hslash$.

In the present article, we focus on the rigorous computation of the first
order of the transition amplitudes \eqref{eq.int.1} given by the limit $\hslash\to
0$. In this topic various quantum field theories can be considered. For
relevance and convenience, we focus on the nonlocal Yukawa particle-field
model. In particular, the spectral and scattering theory for such system is
thoroughly investigated \citep[see, \emph{e.g.},][]{derezinski1999rmp,
  MR1809881, MR1878286}.

The Yukawa theory describes the nuclear force as the interaction of a Dirac
(fermion) field $\psi$ with a Klein-Gordon (boson) field $\phi$ to be given by the
expression
\begin{equation}
  g\, \int_{\R^d} \overline{\psi}(x)  \phi(x)  \psi(x) dx\,,
\end{equation}
where $g$ is the corresponding coupling constant. Considering
non-relativistic nucleons and fixing their number in the above theory, one
obtains the so-called Nelson model, that was studied in the land-marking
article \cite{MR0175537} by Edward Nelson.  The Hamiltonian that we shall
consider here is a reduction of the original Yukawa theory, in the sense that
we impose a nonlocal interaction through an ultraviolet cutoff and a
boson-boson coupling instead of fermion-boson one. Under such simplification,
the first two authors proved that the semiclassical limit $\hslash\to0$ yields the
Schr\"odinger-Klein-Gordon equation, see \cite{MR3059880,ammari2014jsp} and
also \cite{ammari2017sima} where the ultraviolet cutoff is removed. The above
mentioned works concern either stationary solutions or finite time dynamics,
and do not address the problem of large times and scattering theory. In fact,
as far as we know, there are no rigorous results on the $\hslash$-asymptotics of
the scattering and transition amplitudes for the Nelson and Yukawa models.

Let us briefly overview our main contribution. The Hamiltonian describing the
reduced Yukawa theory considered here is given by
\begin{eqnarray*} H_\hslash = H_\hslash^0+\int_{\R^d} \psi_\hslash^*(x) (a_\hslash^*(\lambda_x)+a_\hslash(\lambda_x))\, \psi_\hslash(x)
  \;\mathrm{d}x \,,
\end{eqnarray*} with the non-interacting Hamiltonian defined as
\begin{eqnarray*} H_\hslash^0= \int_{\mathbb{R}^d} \, \psi_{\hslash}^*(x) (-\Delta_x+V(x))\, \psi_{\hslash}(x)\,
  \mathrm{d} x + \int_{\mathbb{R}^d} \, a_{\hslash}^*(k) \omega(k)\,a_{\hslash}(k)\, \mathrm{d} k\,.
\end{eqnarray*} Here $\psi_\hslash^\sharp$ and $a_\hslash^\sharp$ denote respectively the particle and
meson creation-annihilation operators while
$$\omega(k)=\sqrt{k^2+m^2}\,, \; m>0\,,$$
is the meson dispersion relation and $\lambda_x$ is a form factor (see Section
\ref{sec.1} for more details). Then the asymptotic fields are defined as
follows,
\begin{equation} a^{\pm,\sharp}_\hslash(\eta) = \lim_{t\to \pm\infty} e^{-i\frac{t}{\hslash}H_{\hslash}} \;
  e^{i\frac{t}{\hslash}H_{\hslash}^0} \; a_\hslash^\sharp(\eta) \; e^{-i\frac{t}{\hslash}H_{\hslash}^0}\;
  e^{i\frac{t}{\hslash}H_{\hslash}}\,,
\end{equation} and they are known to exist for a dense subset of functions
$\eta\in L^2(\R^d)$, see
\cite{hoegh-krohn1968jmp, hoegh-krohn1969jmp1, MR1326139,
  derezinski1999rmp}.
Our first main result gives the semiclassical limit of the transition
amplitudes under some natural assumptions,
\begin{equation}
  \label{eq:intro1}
  \lim_{\hslash\to 0}   \langle f^\alpha_\hslash, f^\beta_\hslash\rangle=
  \int_{L^2(\R^d)\oplus L^2(\R^d)} \prod_{j} \langle \alpha_j, \Lambda^{+}(u,z)\rangle_{L^2(\R^d)} \; \prod_{k} \langle
  \Lambda^{+}(u,z), \beta_k\rangle_{L^2(\R^d)} \; d\mu_0(u,z)\,,
\end{equation}
and
\begin{equation}
  \label{eq:intro2}
  \lim_{\hslash\to 0}   \langle i^\alpha_\hslash, i^\beta_\hslash\rangle\,,=
  \int_{L^2(\R^d)\oplus L^2(\R^d)} \prod_{j} \langle \alpha_j, \Lambda^{-}(u,z)\rangle_{L^2(\R^d)} \; \prod_{k} \langle
  \Lambda^{-}(u,z), \beta_k\rangle_{L^2(\R^d)} \; d\mu_0(u,z)\,,
\end{equation}
where $\Lambda^{\pm}$ are the in-out-going wave operators of the nonlinear
Schr\"odinger-Klein-Gordon equation \eqref{eq:skg}, constructed in Section
\ref{sec:3} (Definition \ref{def.waveop}) and $\mu_0$ is a semiclassical or
Wigner measure of the family of states $\{\Omega_\hslash\}_{\hslash \in(0,1)}$, which is a Borel
probability measure over the phase-space $L^2(\mathbb{R}^d)\oplus L^2 (\mathbb{R}^d)$. Such result
is stated throughout Theorem \ref{thm:1}, Proposition \ref{prop:segalfield},
and Corollaries \ref{cor:asymptcreat} and \ref{cor:1}.  We further show that
if $\{\Omega_\hslash\}_{\hslash \in(0,1)}$ are the ground states of the interacting system, then
the semiclassical measure $\mu_0$ concentrates on the lowest classical energy
level of the Schr\"odinger-Klein-Gordon system. More generally, if we replace
$\{\Omega_\hslash\}_{\hslash \in(0,1)}$ by excited or bound sates of the quantum Yukawa system,
we then show that $\mu_0$ concentrates on a set of classical asymptotic
radiationless states, from which no radiation is coming out or in (see
Theorem \ref{prop:measvacuum}). The later notion is similar to that of
trapped trajectories in finite dimensional semiclassical analysis, and here
it is defined by means of the wave operators $\Lambda^\pm$, through the condition
$$
\Lambda^\pm(u_0,z_0)=0\,.
$$
Our final result shows that the energy functional of the
Schr\"odinger-Klein-Gordon system \eqref{eq.skg-eng} admits a unique
minimizer $(u_\delta,z_\delta)$, up to invariance, under the constraint
$\|u\|_{L^2(\R^d)}=\delta$ for $\delta$ sufficiently small. This leads to the
identification of the semiclassical measure $\mu_0$ in Corollary \ref{cor:5}
when $\{\Omega_\hslash\}_{\hslash \in(0,1)}$ are the ground states of $H_\hslash$ restricted to the
space of finite number of nucleons $n_\hslash$ such that $\hslash n_\hslash\to \delta$.

\noindent The key point in the proofs of the results above is to combine
dispersive and uniform energy estimates with semiclassical analysis and
scattering theory.

\bigskip
\noindent
\emph{Outlook of the article}: The reduced particle-field Yukawa model is
introduced in Section \ref{sec.1}.  A technical part of the article lays in
Section \ref{sec.2}, where uniform energy and dispersive estimates are given
for both the classical S-KG equation and the quantum Yukawa model.  Then
scattering theory of the Schr\"odinger-Klein-Gordon system is discussed in
Section \ref{sec:3}, where we construct the classical out-in-coming wave
operators.  Our main results on transition amplitudes are proved in Section
\ref{sec:4}.  Precise semiclassical concentration properties of asymptotic
vacuum and bound states are given in Section \ref{sec:5}. In the same
section, uniqueness of minimizers for the S-KG energy functional and
consequences of the asymptotic completeness for the quantum Yukawa model are
addressed.  A list of perspectives and open problems is provided in Section
\ref{sec.op}.

\medskip
\noindent
\emph{Acknowledgments}: M.F.\ has been supported by the European Research
Council (ERC) under the European Union’s Horizon 2020 research and innovation
programme (ERC CoG UniCoSM, grant agreement n.724939). M.O.\ has been
supported by the Deutsche Forschungsgemeinschaft (DFG, Project-ID 258734477 -
SFB 1173). Concerning his research stay in Rennes, in which part of this
project was carried out, M.O.\ would also like to thank LYSM (Laboratoire
Ypatia de Sciences Mathématiques) for the financial support, and IRMAR
(Institut de recherche mathématique de Rennes) for the hospitality.

\section{Yukawa model, semiclassical analysis and asymptotic meson fields}
\label{sec.1}

\subsection{Yukawa's reduced Hamiltonian}
The reduced Yukawa model is composed by two subsystems mutually interacting:
the first one being an arbitrary number of boson particles, considered as a
non-relativistic field obeying Bose-Einstein statistics, while the second one
is a scalar meson field. The Hilbert space of the fully interacting dynamical
system is
\[
  \mathscr{H} := \Gamma_s(\mathfrak{H}) \otimes \Gamma_s(\mathfrak{H}) \cong
  \Gamma_s(\mathscr{Z})\,,
\]
where $\Gamma_s(\cdot )$ denotes the symmetric Fock space, $\mathfrak{H} = L^2(\mathbb{R}^d)$,
and $\mathscr{Z} :=\mathfrak{H} \oplus \mathfrak{H}$. The latter shall be
interpreted as the classical phase-space for the particle-field system. For
convenience, we denote the sector with $n$ non-relativistic particles by
\begin{equation*}
  \mathscr{H}^{(n)}:= L_s^2(\mathbb{R}^{dn}) \otimes \Gamma_s(\mathfrak{H})\,,
\end{equation*}
and the sector with $n$ particles and $m$ mesons by
\begin{equation*}
  \mathscr{H}^{(n,m)} = L_s^2(\mathbb{R}^{dn}) \otimes L_s^2(\mathbb{R}^{dm})\,.
\end{equation*}
It follows that
\[
  \mathscr{H} = \bigoplus_{n \in \mathbb{N}} \mathscr{H}^{(n)}
  =\bigoplus_{n,m \in \mathbb{N}} \mathscr{H}^{(n,m)}.
\]
Consider now $(\psi_\hslash^*, \psi_\hslash)$ and $(a_\hslash^*, a_\hslash)$ to be the creation and
annihilation operators for the particle and meson fields, respectively.
These are $\hslash$-dependent couples of quantum observables, both satisfying the
canonical commutation relations:
\begin{align*} [\psi_{\hslash}(x), \psi^*_{\hslash}(y)]= \hslash\,
  \delta(x-y),\qquad [\, a_{\hslash}(k)\,, a^*_{\hslash}(k')]= \hslash\,
  \delta(k-k'), \qquad [a^{\natural}_{\hslash}(k),
  \psi_{\hslash}^{\natural}(x)] = 0\;.
\end{align*}
The creation and annihilation operators above are in fact considered
``operator-valued distributions'', in the sense that for any $\eta\in L^2(\mathbb{R}^d)$,
\begin{align*}
  &\psi_{\hslash}(\eta)= \int_{\mathbb{R}^d}^{}\bar{\eta}(x)\psi_{\hslash}(x)  \mathrm{d}x\; ,  &\psi_{\hslash}^{*}(\eta)= \int_{\mathbb{R}^d}^{}\eta(x)\psi^{*}_{\hslash}(x)  \mathrm{d}x\;;\\
  &a_{\hslash}(\eta)= \int_{\mathbb{R}^d}^{}\bar{\eta}(k)a_{\hslash}(k)  \mathrm{d}k\; , &a_{\hslash}^{*}(\eta)= \int_{\mathbb{R}^d}^{}\eta(k)a^{*}_{\hslash}(k)  \mathrm{d}k\;;
\end{align*}
are closed operators on $\mathscr{H}$, respectively one adjoint to the
other. Recall as well that $\psi^{\natural}_{\hslash}$ and $a^{\natural}_{\hslash}$ can be written in
terms of the standard $\hslash$-independent creation and annihilation operators
\cite[see, \emph{e.g.},][for an introductory reference]{MR3060648}:
$$
\psi_{\hslash}^\natural(\cdot)=\sqrt{\hslash} \;\psi^\natural(\cdot)\; ;
a^{\natural}_{\hslash}(\cdot)=\sqrt{\hslash}\; a^{\natural}(\cdot)\;.
$$
Given any self-adjoint operator $A $ on $\mathfrak{H}$, we denote its second
quantizations by
\begin{align*}
  \mathrm{d} \Gamma_{\hslash}^{(1)} (A) := \int_{\mathbb{R}^d} \mathrm{d} x \, \psi_{\hslash}^*(x) A\, \psi_{\hslash}(x)\quad , \quad \mathrm{d} \Gamma_{\hslash}^{(2)} (A) := \int_{\mathbb{R}^d} \mathrm{d} k \, a_{\hslash}^*(k) A \,a_{\hslash}(k)\,.
\end{align*}
In particular, the rescaled number operators are respectively
$$
N_1 := \mathrm{d} \Gamma_{\hslash}^{(1)} (1) \otimes \mathds{1}\,, \quad N_2
:= \mathds{1} \otimes \mathrm{d} \Gamma_{\hslash}^{(2)} (1)\,,
$$
and the total rescaled number operator is
 $$
 N = N_1 + N_2\,.
 $$
 Remark that $N_1$ and $N_2$ count $\hslash$-times the number of particles in
 each respective sector.

 With the above notations, the particle-field Yukawa Hamiltonian is given by
 \begin{equation}
   \label{eq:hamint}
   H_{\hslash} := H_\hslash^0 + H_\hslash^I\,,
 \end{equation}
 such that
 \begin{equation*}
   \begin{aligned}
     H_\hslash^0 &:= \mathrm{d} \Gamma_{\hslash}^{(1)} (-\Delta + V) \otimes \mathds{1} + \mathds{1} \otimes \mathrm{d} \Gamma_{\hslash}^{(2)} (\omega)\,, \\
     H_\hslash^I &:= \mathrm{d} \Gamma_{\hslash}^{(1)} (\phi_{\hslash}(\lambda_{x}))\,,
   \end{aligned}
 \end{equation*}
 and
 \begin{equation}
   \label{eq:phidef}
   \phi_{\hslash}(\lambda_x) := a_{\hslash}(\lambda_x) + a^{*}_{\hslash}(\lambda_x)\,.
 \end{equation}

  \begin{remark}
   \label{rem:1}
   The Yukawa Hamiltonian $H_{\hslash}$ strongly commutes with the nucleon
   number operator $N_1$ \citep[see, \emph{e.g.},][Appendix A]{MR3379490},
   therefore it can be fibered to each $\mathscr{H}^{(n)}$:
   \begin{equation*}
     H_{\hslash}=\bigoplus_{n\in \mathbb{N}} H_{\hslash}^{(n)}=\bigoplus_{n\in \mathbb{N}}  \Bigl(H_\hslash^{0,(n)}+H_\hslash^{I,(n)}\Bigr)\; .
   \end{equation*}
 \end{remark}
 
 The above potential $V$, dispersion relation $\omega$, and form factor $\lambda_x(\cdot)$
 need to satisfy the following regularity assumptions.
 
 \begin{hyp}
   \label{hyp:1}
   $\phantom{i}$
   \begin{itemize}
   \item $V : \mathbb{R}^d \rightarrow \mathbb{R}^+$ is locally
     integrable. In addition, there exist $C_0 >0$ and $\nu > 0 $ such that:
     \begin{equation}\label{hyp:pot}\tag{A1}
       V(x) \geq C_0\, \langle 	x\rangle^{1+\nu}\,.
     \end{equation}
     Here $\langle\,\cdot\,\rangle_{}=(1+\lvert \,\cdot\, \rvert_{}^2)^{\frac{1}{2}}$ stands for the so-called
     Japanese bracket.

   \item $\omega$ is the dispersion of a scalar relativistic field of mass
     $m$:
     \begin{equation}\label{hyp:omega}\tag{A2}
       \omega(k) = \sqrt{k^2 + m^2}, \quad m > 0\,.
     \end{equation}

   \item There exists a cut-off function $\chi \in C^{\infty}_0(\mathbb{R}^d)$ such that
     \begin{equation}\label{hyp:lambda}\tag{A3}
       \lambda_x(k) = e^{ik\cdot x} \omega^{-1/2}(k) \chi(k)\,.
     \end{equation}
   \end{itemize}
 \end{hyp}
 From now on, we will always assume implicitly
 \eqref{hyp:pot}-\eqref{hyp:lambda} to be satisfied.

\bigskip
 
 \subsection{Semiclassical limit and Wigner measures}
 \label{subsec:wig}
 The classical limit procedure, known also as Bohr's correspondence
 principle, can be interpreted as the approximation of the behavior of a
 quantum system by its classical counterpart in a certain effective
 regime. There are several ways to give a mathematical foundation to this
 principle. Among them, one can mention the methods based on coherent states,
 path integration, Schwinger-Dyson expansion and Wigner measures. The latter
 is the one that we shall follow here. In fact, the Wigner measure approach
 initiated in \cite{ammari2008ahp} for quantum field and many-body theories
 proved to be quite effective \citep[see, \emph{e.g.},][and references
 therein contained]{MR3907740, ammari2014jsp, ammari2017sima, MR3379490}.
 More precisely, consider a family of normalized vectors $\{\Psi_{\hslash}\}_{\hslash \in
   (0,1)} \subseteq \mathscr{H}$. The main point consists in studying the cluster
 points of the expectations
 \[
   \left\langle \Psi_{\hslash},W_\hslash(\eta) \Psi_{\hslash}
   \right\rangle_{\mathscr{H}}\,,
 \]
 where $\eta=\eta_1\oplus \eta_2\in \mathscr{Z}$ and $W_\hslash(\eta)$ is the \emph{Weyl
   operator}. Recall that the Weyl operator is a unitary operator on
 $\mathscr{H}$, defined by
 \begin{equation}
   \label{eq:weyldef}
   W_\hslash(\eta)=  W_{\hslash}^{(1)}(\eta_1)\otimes W_{\hslash}^{(2)}(\eta_2)\equiv e^{\frac{i}{\sqrt{2}} \bigl(\psi_\hslash^{*} (\eta _1)+\psi_\hslash(\eta _1)\bigr)+\frac{i}{\sqrt{2}}\bigl(a_\hslash^{*}(\eta _2)+a_\hslash(\eta _2)\bigr) }\;.
 \end{equation}
 As proved in \citep[Theorem 6.2]{ammari2008ahp}, if there exists $C>0$ such
 that
 \begin{equation}
   \label{number-cond}
   \forall \hslash\in(0,1), \quad \langle \Psi_{\hslash}  , N \Psi_{\hslash}\rangle_{\mathscr{H}} \leq C\,,
 \end{equation}
 then there is at least one sequence $(\hslash_n)_{n\in\mathbb{N}}$, $\hslash_n\to 0$, and a Borel
 probability measure $\mu$ on $\mathscr{Z}$ such that for all $\eta\in\mathscr{Z}$,
 \begin{equation}
   \label{wigner}
   \lim_{n\to\infty} \left\langle \Psi_{\hslash_n},W_{\hslash_n}(\eta) \Psi_{\hslash_n} \right\rangle_{\mathscr{H}}=  \int_{\mathscr{Z}} e^{\sqrt{2}i \Re\langle \eta_1\oplus \eta_2, u\oplus z\rangle_{\mathscr{Z}}} \,d\mu(u,z)\,.
 \end{equation}
 The expectation on the left-hand side is the Fourier-Wigner transform of a
 (bosonic) quantum state while the average on the right-hand side is the
 (inverse) Fourier transform of a measure $\mu$. Thus the convergence given
 in \eqref{wigner} is interpreted as a semiclassical limit emphasizing the
 relationship between quantum field generating functionals and classical
 characteristic functions. The formal infinitesimal version of the limit
 \eqref{wigner} yields the the convergence of quantum correlation functions
 towards the classical ones.  Symbolically, we denote this convergence by
 $\Psi_{\hslash_n}\rightharpoondown \mu$. The probability measures $\mu$ such
 that $\Psi_{\hslash_k}\rightharpoondown \mu$ for some sequence $\hslash_k\to
 0$, are called Wigner or semiclassical measures of the family
 $\{\Psi_{\hslash}\}_{\hslash\in (0,1)}$. The set of all Wigner measures
 associated with the family $\{\Psi_{\hslash}\}_{\hslash\in (0,1)}$ is
 denoted by
 \[
   \mathscr{M}\bigl(\Psi_{\hslash}; \hslash\in(0,1)\bigr)\,.
 \]

 From a dynamical standpoint, it is quite relevant to understand whether a
 classical dynamics in the phase space $\mathscr{Z}$ is a good approximation
 in the semiclassical regime to the microscopic dynamics of quantum
 fields. This can be rigorously done giving a characterization of the Wigner
 measures associated to the time-evolved vectors $\Psi_\hslash(t)= e^{-i
   \frac{t}{\hslash} H_\hslash} \Psi_\hslash$. Indeed, it is proved in
 \citep[Theorem 1.1]{ammari2014jsp} for the Yukawa model that if
 $\mathscr{M}\bigl(\Psi_{\hslash_n}; n\in \mathbb{N}\bigr)=\{\mu\}$ at time
 zero, then for all later or previous times
 \begin{equation*}
   \mathscr{M}\bigl(\Psi_{\hslash_n}(t); n\in \mathbb{N}\bigr)=\{(\Phi_t) \, _{*}\, \mu\}\; ,
 \end{equation*}
 where $\Phi_t$ is the nonlinear Hamiltonian flow of the
 Schr\"odinger-Klein-Gordon system given in \eqref{eq:skg} and $(\Phi_t) \,
 _{*}\, \mu$ is the push-forward measure defined for any Borel set $B$ in
 $\mathscr{Z}$ by
 $$
 (\Phi_t) \, _{*}\, \mu(B)=\mu\big( (\Phi_t)^{-1}(B)\big)\,.
 $$
 We refer the reader to Section \ref{sec:3} for more details on the S-KG
 equation, and to our previous papers \citep{ammari2008ahp, ammari2014jsp} on
 Wigner measures and semiclassical limits.

\subsection{Asymptotic meson fields}
Knowing the behavior of the quantum system in the semiclassical limit at
fixed times $t\in \mathbb{R}$, it is natural to ask whether the approximation
by the classical dynamics holds true for asymptotically long times. In order
to deal with such a question, it is necessary to introduce the concept of
\emph{asymptotic fields}
\citep[see][]{hoegh-krohn1968jmp,hoegh-krohn1969jmp1,hoegh-krohn1969jmp2,derezinski1999rmp}.

\begin{definition}
  The asymptotic Weyl operators for the meson field, denoted by
  $W^{\pm}_{\hslash}(\eta)$, are defined for any $\eta \in \mathfrak{H}$ by
  \begin{equation}
    \begin{aligned}
      W^{\pm}_{\hslash} (\eta) &:= \slim_{t \rightarrow \pm \infty} e^{i\frac{t}{\hslash}H_{\hslash}}e^{-i \frac{t}{\hslash} H_{\hslash}^0} \,W^{(2)}_{\hslash}(\eta)e^{i \frac{t}{\hslash} H_{\hslash}^0} \,e^{-i\frac{t}{\hslash}H_{\hslash}} \\
      &=\slim_{t \rightarrow \pm \infty} e^{i\frac{t}{\hslash}H_{\hslash}}W^{(2)}_{\hslash}(\eta_t)e^{-i\frac{t}{\hslash}H_{\hslash}}\;,
    \end{aligned}
  \end{equation}
  and exist as $\mathscr{H}$-strong limits. In the above equation,
  $\eta_t=e^{-it \omega} \eta$.
\end{definition}
The existence of the asymptotic Weyl operators for the relativistic meson
field in the Yukawa model was proved firstly in \citep{hoegh-krohn1968jmp},
see also \citep{derezinski1999rmp} for additional details. The asymptotic
Weyl operators are the fundamental objects in the study of the scattering
properties of the Yukawa model, playing a role analogous to wave operators in
the scattering theory of Schrödinger operators. They satisfy indeed the Weyl
relations
\begin{equation*}
  W^{\pm}_{\hslash} (\eta)W^{\pm}_{\hslash} (\eta')=e^{\frac{i}{2}\Im \langle \eta  , \eta' \rangle_{\mathfrak{H}}}W^{\pm}_{\hslash} (\eta+\eta')\;.
\end{equation*}
The collection $\{W^{\pm}_{\hslash}(\eta),\eta\in \mathfrak{H}\}$ generate thus, in the
algebraic sense, a representation of the time-zero canonical commutation
relations for the relativistic free field. It is also possible to define the
associated asymptotic fields $\phi^{\pm}_{\hslash}(\eta)$, and creation-annihilation
operators $a^{\pm,\natural}_{\hslash}(\eta)$.  In order to do that, observe that according to
\cref{rem:1}, the asymptotic Weyl operators $W^{\pm}_{\hslash}(\eta)$ strongly commute
with $N_1$, hence they can be fibered on each $\mathscr{H}^{(n)}$ ,
\begin{equation}
  \label{eq:1}
  W_{\hslash}^{\pm}(\eta)=\bigoplus_{n\in \mathbb{N}} W^{\pm}_{\hslash}(\eta)^{(n)}\,.
\end{equation}
Moreover, notice that the maps
\begin{equation*}
  \mathfrak{H}\ni \eta \mapsto W^{\pm}_{\hslash}(\eta)^{(n)}= W^{\pm}_{\hslash}(\eta)_{|\mathscr{H}^{(n)}}\in \mathscr{B}(\mathscr{H}^{(n)})
\end{equation*}
are continuous with respect to the strong operator topology, and the maps
\begin{equation*}
  \mathfrak{H}\ni \eta \mapsto W^{\pm}_{\hslash}(\eta)^{(n)}(H_{\hslash}^{(n)}+i)^{-\varepsilon} \in \mathscr{B}(\mathscr{H}^{(n)})
\end{equation*}
are continuous, for any $\varepsilon>0$, with respect to the uniform operator
topology \citep{derezinski1999rmp}. Hence, by Stone's theorem, it is possible
to define, for any $\eta\in \mathfrak{H}$, $\phi^{\pm}_{\hslash}(\eta)^{(n)}$
as the generators of the asymptotic Weyl operators
\begin{equation*}
  W^{\pm}_{\hslash}(\eta)^{(n)}=e^{\frac{i}{\sqrt{2}}\phi^{\pm}_{\hslash}(\eta)^{(n)}}\; .
\end{equation*}
Additionally, one concludes for instance using \citep[][Proposition
A1]{MR3379490} that the operators
\begin{equation}
  \label{eq:22}
  \phi_{\hslash}^{\pm}(\eta)=\bigoplus_{n\in \mathbb{N}} \phi^{\pm}_{\hslash}(\eta)^{(n)}\;.
\end{equation}
are self-adjoint and satisfy for all $\eta\in \mathfrak{H}$,
$$
W^{\pm}_{\hslash}(\eta) =e^{\frac{i}{\sqrt{2}}\phi^{\pm}_{\hslash}(\eta)}\; .
$$

Finally, note that the operators $\{\phi^{\pm}_{\hslash}(\eta)^{(n)},\eta\in
\mathfrak{H}\}$ have a common dense domain, since the following inequality
holds for any $\eta\in \mathfrak{H}$ \citep[Theorem 5.2]{derezinski1999rmp}:
\begin{equation}
  \label{eq:4}
  \Bigl\lVert \phi^{\pm}_{\hslash}(\eta)^{(n)} (H_{\hslash}^{(n)}+i)^{-\frac{1}{2}} \Bigr\rVert_{\mathscr{B}(\mathscr{H}^{(n)})}^{}\leq C\lVert \eta  \rVert_{\mathfrak{H}}^{}\; .
\end{equation}
Therefore the form domain $\mathscr{Q}(H_{\hslash}^{(n)})\subset
\mathscr{D}(\phi^{\pm}_{\hslash}(\eta)^{(n)})$ for all $\eta\in
\mathfrak{H}$. The creation and annihilation operators
\begin{align}
  \label{eq:2}
  &a^{\pm}_{\hslash}(\eta)^{(n)}= \frac{1}{2} \bigg(\phi^{\pm}_{\hslash}(\eta)^{(n)} +i \phi^{\pm}_{\hslash}(i\eta)^{(n)}\bigg)\; ,\\
  \label{eq:3}
  &a^{\pm,*}_{\hslash}(\eta)^{(n)}= \frac{1}{2} \bigg(\phi^{\pm}_{\hslash}(\eta)^{(n)} -i \phi^{\pm}_{\hslash}(i\eta)^{(n)}\bigg)\; ,
\end{align}
are then densely defined and closed on
$\mathscr{D}(\phi^{\pm}_{\hslash}(\eta)^{(n)})\cap
\mathscr{D}(\phi^{\pm}_{\hslash}(i\eta)^{(n)})\supset
\mathscr{Q}(H_{\hslash}^{(n)})$ for all $\eta\in \mathfrak{H}$ \citep[Theorem
5.3]{derezinski1999rmp}.  Henceforth, one defines
\begin{equation}
  \label{eq:14}
  a_{\hslash}^{\pm,\natural}(\eta)=\bigoplus_{n\in \mathbb{N}} a^{\pm,\natural}_{\hslash}(\eta)^{(n)}\;,
\end{equation}
as closed operators on the fibred Hilbert space $\mathscr{H}$.

We conclude this section by introducing the notion of \emph{asymptotic
  completeness} for the Yukawa model. The $n$-nucleon asymptotic vacuum
spaces are defined respectively as
\begin{equation}
  \label{eq:19}
  \mathscr{K}^{\pm,(n)}_{\hslash}=\Bigl\{ \psi^{(n)}\in \mathscr{H}^{(n)}\;\vert\; \forall \eta\in \mathfrak{H}\,,\, a^{\pm}_{\hslash}(\eta)^{(n)}\psi^{(n)}=0 \Bigr\}\; .
\end{equation}
Physically, the spaces $\mathscr{K}^{\pm,(n)}_{\hslash}$ can be thought as
the configurations in which only the $n$ asymptotic (``dressed'') nucleons
are present, for there is no asymptotically free meson in this
case. Asymptotic completeness is formulated as a property of asymptotic
vacua. Indeed, the Yukawa theory is said to be asymptotically complete iff
(each) $\mathscr{K}^{\pm,(n)}_{\hslash}$ equals the spectral subspace of
bound states of $H_{\hslash}^{(n)}$, {\it i.e.}
\begin{equation}
  \label{eq:20}
  \mathscr{K}^{\pm,(n)}_{\hslash}= \mathds{1}_{\mathrm{pp}}(H_{\hslash}^{(n)}) \mathscr{H}^{(n)}\; ,
\end{equation}
where $\mathds{1}_{\mathrm{pp}}(H_{\hslash}^{(n)})$ denotes the spectral
projection over the pure point spectrum
$\sigma_{\mathrm{pp}}(H_\hslash^{(n)})$.  This guarantees in particular that
for each $n\in \mathbb{N}$,
$\mathscr{K}^{+,(n)}_{\hslash}=\mathscr{K}^{-,(n)}_{\hslash}$ (future and
past asymptotic vacua coincide), and that the asymptotic vacua are bound
states of the Yukawa Hamiltonian.

\section{Technical estimates}
\label{sec.2}
In this section we collect some important technical estimates that are used
to prove our main results.

\subsection{Energy estimates}
\label{subsec:est-enr} We recall here some uniform estimates for the
Hamiltonian $H^{(n)}_\hslash$. For each fixed $\hslash$, such estimates are
well-known; however, for our analysis of the limit $\hslash\to 0$, uniformity
with respect to $\hslash$ has to be kept into account.

\begin{lemma}\label{lem:dgammaomega}
  There exists $C>0$ such that for all $\hslash\in(0,1)$,
  \[\|(\mathrm{d} \Gamma_{\hslash}^{(2)}(\omega) + 1)^{-1/2}
    a^{\natural}_{\hslash}(\lambda_{x})\|_{\mathscr{B}(\mathscr{H}^{(1)})} \leq
    C.\]
\end{lemma}
\begin{proof}
  Let us start with the annihilation operator, and let $\Psi \in
  \mathscr{H}^{(1,n)}$. Then,
  \begin{equation*}
    \begin{split}
      \| a_{\hslash}(\lambda_{x}) \Psi\|_{\mathscr{H}^{(1,n-1)}}^2 = \int_{\mathbb{R}^d \times \mathbb{R}^{d(n-1)}} \mathrm{d} x \,\mathrm{d} k_2 \ldots \mathrm{d} k_n  \, \hslash n \,\left| \int_{\mathbb{R}^d} \mathrm{d} k_1\, \overline{\lambda_x(k_1)} \, \Psi(x, k_1, \ldots, k_n) \right|^2\\
      = \int_{\mathbb{R}^d \times \mathbb{R}^{d(n-1)}} \mathrm{d} x \,\mathrm{d} k_2 \ldots \mathrm{d} k_n  \, \hslash n \,\left| \int_{\mathbb{R}^d} \mathrm{d} k_1\, \omega(k_1)^{-\frac{1}{2}}\overline{\lambda_x(k_1)} \, \omega(k_1)^{\frac{1}{2}}\Psi(x, k_1, \ldots, k_n) \right|^2\\
      \leq \int_{\mathbb{R}^d \times \mathbb{R}^{d(n-1)}} \mathrm{d} x \,\mathrm{d} k_1 \ldots \mathrm{d} k_n\, \bigl\lVert \omega^{-\frac{1}{2}}\lambda_x  \bigr\rVert_{\mathfrak{H}}^2 \hslash n\, \omega(k_1)\, \bigl\lvert \Psi(x, k_1, \ldots, k_n)  \bigr\rvert_{}^2\\
      \leq \|\omega^{-\frac{1}{2}}\lambda_{(\cdot)}\|^2_{L^{\infty}(\mathbb{R}^{d}; \mathfrak{H})} \lVert \mathrm{d}\Gamma^{(2)}_{\hslash}(\omega)^{\frac{1}{2}}\Psi  \rVert_{\mathscr{H}^{(1,n)}}^{2}\leq \frac{1}{m} \|\lambda_{(\cdot)}\|^2_{L^{\infty}(\mathbb{R}^{d}; \mathfrak{H})} \lVert\mathrm{d}\Gamma^{(2)}_{\hslash}(\omega)^{\frac{1}{2}}\Psi  \rVert_{\mathscr{H}^{(1,n)}}^{2}\; .
    \end{split}
  \end{equation*}
  The proof for the creation operator makes use of the canonical commutation
  relations, of the result for the annihilation operator, and of the fact
  that $\hslash\in (0,1)$:
  \begin{equation*}
    \begin{split}
      \| a_{\hslash}^{*}(\lambda_{x}) \Psi\|_{\mathscr{H}^{(1)}}^2=\| a_{\hslash}(\lambda_{x}) \Psi\|_{\mathscr{H}^{(1)}}^2+ \langle \Psi  , \hslash\lVert \lambda_x  \rVert_{\mathfrak{H}}^2\Psi \rangle_{\mathscr{H}^{(1)}}\\\leq (1+\tfrac{1}{m})\|\lambda_{(\cdot)}\|^2_{L^{\infty}(\mathbb{R}^{d}; \mathfrak{H})}\bigl\langle \Psi  ,  \bigl(\mathrm{d}\Gamma^{(2)}_{\hslash}(\omega)+1\bigr)\Psi\bigr\rangle_{\mathscr{H}^{(1)}}\; .
    \end{split}
  \end{equation*}
\end{proof}

\begin{lemma}\label{lem:KR.est}
  Let $0<\alpha<\beta$ given. Then there exist $b,C>0$ such that for any
  $n\in\mathbb{N}$ and $\hslash\in(0,1)$ satisfying $n\hslash\in
  (\alpha,\beta)$, one has:
  \begin{gather}
    \label{eq:6}
    \| H_\hslash^{(n) } \; (H_\hslash^{0,(n)}+1)^{-1}\|_{\mathscr{B}(\mathscr{H})}\leq C\; ,\\
    \label{eq:7}
    \| H_\hslash^{0,(n) } \; (H_\hslash^{(n)}+b)^{-1}\|_{\mathscr{B}(\mathscr{H})}\leq C\; .
  \end{gather}
\end{lemma}
\begin{proof}
  Thanks to the estimate of \cref{lem:dgammaomega}, there exists $C>0$ such
  that
  \begin{equation}
    \label{eq:25}
    \Bigl\|\hslash \sum_{j=1}^n \phi_\hslash(\lambda_{x_j}) \; (H_\hslash^{0,(n)}+1)^{-1/2}\Bigr\|_{\mathscr{B}(\mathscr{H})} \leq C\,,
  \end{equation}
  uniformly in $n$ and $\hslash$ such that $n\hslash\in (\alpha,\beta)$. This
  proves \eqref{eq:6}, since
  $$ H_\hslash^{(n) } = H_\hslash^{0,(n) }+H_\hslash^{I,(n) }=
  H_\hslash^{0,(n) }+ \hslash \sum_{j=1}^n \phi_\hslash(\lambda_{x_j})\,.$$
  In particular, it follows that for any $\gamma\in(0,1)$ there exists $b >0$
  such that, for any $\varphi\in \mathscr{D}(H_\hslash^{0,(n)})$,
  \begin{equation}
    \label{eq.KR}
    \| H_\hslash^{I, (n)} \varphi\| \leq \gamma \,\| H_\hslash^{0,(n)} \varphi\|+ b\, \| \varphi\|\,,
  \end{equation}
  uniformly in $n$ and $\hslash$ as before.  As a consequence of the
  Kato-Rellich Theorem, $H_\hslash^{(n)}$ is self-adjoint on
  $\mathscr{D}(H_\hslash^{0,(n)})$, uniformly bounded from below and
  \eqref{eq:7} holds true for some $b>0$ and $C>0$.
\end{proof}
Concerning the total Hamiltonian $H_{\hslash}$, the following uniform
inequalities are very useful.
\begin{lemma}\label{lem:formcontrol}
  There exist $ c_1, c_2, \underline{\alpha},\underline{\beta}>0$ such that
  for all $\hslash \in (0,1)$,
  \begin{gather}
    \label{eq:10}
    \pm H_{\hslash}^I \leq c_1 \,(H_\hslash^0 + N_1^2 + 1)\; ,\\
    \label{eq:11}
    H_\hslash^0 \leq c_2 \, (H_{\hslash} + \underline{\alpha} N_1^2 + \underline{\beta})\; .
  \end{gather}
\end{lemma}
\begin{proof}
  Let $\varphi^{(n)} \in \mathscr{H}^{(n)}$. By symmetry of the function
  $\varphi^{(n)}$, we have
  \begin{eqnarray}
    \label{eq:control1}
    % \begin{aligned}
    \bigl|\langle \varphi^{(n)},H_\hslash^I \varphi^{(n)} \rangle\bigr| &=& \hslash n \bigl|\langle \varphi^{(n)}, \, \phi_{\hslash}(\lambda_{x_1})\varphi^{(n)}\rangle\bigr|	\nonumber\\
    &\leq& \, \| (\mathrm{d} \Gamma_{\hslash}^{(2)}(\omega) + 1)^{1/2}\varphi^{(n)}\|_{\mathscr{H}^{(n)}}\, \|   (\mathrm{d} \Gamma_{\hslash}^{(2)}(\omega) + 1)^{-1/2}\, \hslash n\,\phi_{\hslash}(\lambda_{x_1})\varphi^{(n)}\|_{\mathscr{H}^{(n)}} \nonumber \\
    &\leq&  C \, \| (\mathrm{d} \Gamma_{\hslash}^{(2)}(\omega) + 1)^{1/2}\varphi^{(n)}\|_{\mathscr{H}^{(n)}}\|N_1 \varphi^{(n)}\|_{\mathscr{H}^{(n)}} \; ,
    %\end{aligned}
  \end{eqnarray}
  where \cref{lem:dgammaomega} is used and $C>0$ is a constant independent of
  $n\in\N$.

  Now, to prove \eqref{eq:10} we use \eqref{eq:control1} to show that
  \begin{align*}
    |\langle \varphi^{(n)},H_\hslash^I\varphi^{(n)} \rangle| &\leq  C \, \| (\mathrm{d} \Gamma_{\hslash}^{(2)}(\omega) + 1)^{1/2}\varphi^{(n)}\|_{\mathscr{H}^{(n)}}\|N_1 \varphi^{(n)}\|_{\mathscr{H}^{(n)}} \\
    &\leq C\, \langle  \varphi^{(n)},  \bigl(\mathrm{d} \Gamma_{\hslash}^{(2)}(\omega) + N^2_1 +1\bigr) \varphi^{(n)} \rangle  \\
    &\leq  C \, \langle  \varphi^{(n)},  \bigl(H_\hslash^0+ N^2_1 +1\bigr)\varphi^{(n)} \rangle\; .
  \end{align*}
  Now, since $H_\hslash^I$ and $H_\hslash^0 + N_1^2 + 1$ commute with $N_1$,
  we can sum over $n$ and conclude the proof.

  To prove \eqref{eq:11}, we use again \eqref{eq:control1}, and the
  inequality $2ab \leq \frac{1}{\eta^2} a^2 + \eta^2 b^2$ that yield
  \begin{align*}
    |\langle \varphi^{(n)},H_\hslash^I \varphi^{(n)} \rangle| &\leq  C \, \| (\mathrm{d} \Gamma_{\hslash}^{(2)}(\omega) + 1)^{1/2}\varphi^{(n)}\|_{\mathscr{H}^{(n)}}\|N_1 \varphi^{(n)}\|_{\mathscr{H}^{(n)}}  \\
    &\leq \frac{C}{\eta^2} \langle \varphi^{(n)}, \bigl(\mathrm{d} \Gamma_{\hslash}^{(2)}(\omega) + 1\bigr)\varphi^{(n)}\rangle + C \eta^2 \langle \varphi^{(n)},  N^2_1 \varphi^{(n)}\rangle\; .
  \end{align*}
  Hence,
  \[
    H_\hslash^I \geq - \frac{C}{\eta^2} (\mathrm{d}
    \Gamma_{\hslash}^{(2)}(\omega) + 1) - C\eta^2 N^2_1\,.
  \]
  Finally, adding $H_\hslash^0$ on both sides one obtains
  \[ H_{\hslash} + C\eta^2 N^2_1 \geq \mathrm{d}
    \Gamma_{\hslash}^{(1)}(-\Delta + V) + \left( 1- \frac{C}{\eta^2}\right)
    \mathrm{d} \Gamma_{\hslash}^{(2)}(\omega) -\frac{C}{\eta^2}. \]Now
  choosing $\eta >0$ such that $C\eta^{-2} < 1$, it follows that
  \[(1- C\eta^{-2}) \, H_\hslash^0 \leq H_{\hslash} + C\eta^2 N^2_1 +
    C\eta^{-2}\,.\]
\end{proof}

The regularity property below is a straightforward consequence of the above
Lemma \ref{lem:formcontrol}. Beforehand, denote
\begin{gather}
  \label{eq:12}
  \Psi_{\hslash}(t) = e^{-i\frac{t}{\hslash} H_{\hslash}} \Psi_{\hslash}\; ,\\
  \label{eq:13}
  S=H_{\hslash}^0+N_1^2+1\; .
\end{gather}

\begin{corollary}
  \label{lem:controltime}
  Assume that $\{\Psi_{\hslash}\}_{\hslash \in (0,1)} \subseteq \mathscr{H}$
  is a family of normalized vectors satisfying for some constant $C>0$ and
  all $\hslash\in(0,1)$,
  \begin{equation*}
    \left\langle \Psi_{\hslash}, S \,\Psi_{\hslash} \right\rangle \leq C.
  \end{equation*}
  Then there exists a constant $c_3>0$ such that, uniformly in
  $\hslash\in(0,1)$ and in $t \in \mathbb{R}$,
  \begin{equation}\label{eq:controltime}
    \left\langle \Psi_{\hslash}(t), S\,\Psi_{\hslash}(t) \right\rangle \leq c_3.
  \end{equation}
\end{corollary}

We conclude this section by proving \emph{uniform energy-number} estimates,
inspired by the ones in \citep[][Section 3.5]{derezinski1999rmp} and which
will be useful in studying correlation functions in
\cref{sec:class-limit-asympt}.
\begin{proposition}
  \label{prop:enrgynumberest}
  Let $0<\alpha<\beta$. For $k,r\in\mathbb{N}$ there exist $c,b>0$ such that
  for any $n\in\mathbb{N}$ and any $\hslash\in(0,1)$ satisfying $n\hslash\in
  (\alpha,\beta)$, the following estimates hold true :
  \begin{enumerate}
  \item\label{item:1} $ \|(N_2+1)^k (H_\hslash^{(n)}+b)^{-k}\|\leq c$.
  \item\label{item:2} $ \|(N_2+1)^{k+r} (H_\hslash^{(n)}+b)^{-k}
    (N_2+1)^{-r}\|\leq c$.
  \item\label{item:3} $ \|H_\hslash^{0,(n)} (N_2+1)^k
    (H_\hslash^{(n)}+b)^{-(k+1)}\|\leq c$.
  \end{enumerate}
\end{proposition}
\begin{proof}
  Throughout the proof, we restrict to $n$ and $\hslash$ satisfying the
  assumption in the proposition. Observe that by \eqref{eq:11} of
  \cref{lem:formcontrol} there exits a constant $b>0$ such that
  $H_\hslash^{(n)}+b \geq 1$.  For simplicity, let us denote
$$
A= N_2+1, \qquad \qquad B= (H_\hslash^{(n)}+b)^{-1}\,,
$$
and $ {\rm ad}^j_A(B)$ the adjoint action defined recursively as
\begin{gather*}
  {\rm ad}^0_A(B)=B \; ,\\
  {\rm ad}^j_A(B)=[A, {\rm ad}^{j-1}_A(B)]\; .
\end{gather*}
Recall the Leibniz's formula
\begin{equation}
  \label{GLF}
  A^k B = \sum_{j=0}^k  {k \choose j} \, {\rm ad}^j_A(B) \, A^{k-j}\,.
\end{equation}
First, observe that for all $j\in\N$
\begin{equation}
  \label{eq:23}
  {\rm ad}^j_A(B)= (i\hslash)^{j} \sum_{p=1}^j c_p \,B  \phi_\hslash (r_{p,1}) \dots B \phi_\hslash(r_{p,p}) B,
\end{equation}
with $r_{p,q}\in\{\hslash\sum_{s=1}^n \lambda_{x_s},i\hslash \sum_{s=1}^n
\lambda_{x_s}\}$ for $q=1,\dots,p,$ and $c_p$ are real coefficients
independent of $\hslash$ and $n$. Such identity is proved by induction on $j$
and using the commutation relations
$$
{\rm ad}^1_A(B)=[A,B]=-i\hslash \;B \, \phi_\hslash\bigg(i\hslash\sum_{s=1}^n
\lambda_{x_s}\bigg) \, B\,,
$$
and
$$
[A,\phi_\hslash(r_{p,q})]=-i\hslash \,\phi_\hslash(ir_{p,q})\,.
$$
In particular, the Leibniz formula \eqref{GLF} make sense as an operator
equality on $\mathscr{D}(A^k)$.

We prove now \ref{item:1} by induction. For $k=1$, the inequality is a
consequence of \cref{lem:KR.est}, and the fact that
$(H_\hslash^{0,(n)}+1)^{-1} N_2$ is bounded uniformly with respect to $n$ and
$\hslash$.  The induction step goes as follows.
\begin{align*}
  A^{k+1} B^{k+1}&= A \left(A^k B\right) B^k \\
  &= A \left( \sum_{j=0}^k  {k \choose j} \, {\rm ad}^j_A(B) \, A^{k-j} \right) B^k \,.
\end{align*}
Using \cref{lem:KR.est} and the inequality \eqref{eq:25}, one shows that
there exists $c>0$ such that for all $p\in\mathbb{N}, q=1,\dots,p,$
\begin{equation}
  \label{eq:comest2}
  \|\phi_\hslash(r_{p,q})  \; B\|_{ \mathscr{B}(\mathscr{H}^{(n)})}  \leq c\,,
\end{equation}
uniformly with respect to $n$ and $\hslash$. Therefore,
\begin{align}
  A^{k+1} B^{k+1}&= A B A^k B^k+
  A B \sum_{j=1}^k  (i\hslash)^j  {k \choose j} \bigg( \sum_{p=1}^j c_p  \,\phi_\hslash(r_{p,1}) \dots B \phi_\hslash(r_{p,p}) B \bigg)
  \; A^{k-j}  B^k \,.
\end{align}
Since all the operators on the right hand side are bounded uniformly with
respect to $n$ and $\hslash$, the proof of \ref{item:1} is completed.

Let us now prove \ref{item:2}. We use the following algebraic operator
identity:
\begin{eqnarray}
  \label{eq.comprod}
  A^{k+r} B^{k} A^{-r}= A^{k+r} B A^{-k-r+1} \dots A^{r+1} B A^{-r}\,.
\end{eqnarray}
The Leibniz's formula and the identity \eqref{eq:23} yield
\begin{eqnarray}
  \label{eq:24}
  A^\gamma B A^{-\gamma+1}= \sum_{j=0}^\gamma (i\hslash)^j  {\gamma \choose j} \bigg( \sum_{p=1}^j c_p  \,B \phi_\hslash(r_{p,1}) \dots B \phi_\hslash(r_{p,p}) B \bigg)
  A^{1-j} \; ,
\end{eqnarray}
for any $\gamma\in\mathbb{N}$. Since all the terms in the right hand side of \eqref{eq:24},
and consequently \eqref{eq.comprod}, are uniformly bounded with respect to
$n$ and $\hslash$, the proof is concluded.

It remains to prove \ref{item:3}. It follows from \ref{item:1}, using
Leibniz's formula. In fact, one has
\begin{align*}
  H_\hslash^{0,(n)} A^k  B^{k+1}&= H_\hslash^{0,(n)} \left( A^k B \right)   B^{k}\\
  &=  H_\hslash^{0,(n)} B A^k B^k+ H_\hslash^{0,(n)} B \sum_{j=1}^k  (i\hslash)^j  {k \choose j} \,\bigg(\sum_{p=1}^j   c_p  \phi_\hslash(r_{p,1}) \cdots B \phi_\hslash(r_{p,p}) B \bigg)   A^{k-j}  B^{k}.
\end{align*}
Hence combining \ref{item:1} with \eqref{eq:comest2} and with the fact that
$H_\hslash^{0,(n)} B$ is uniformly bounded according to \cref{lem:KR.est}, it
follows that all the terms on the right hand side are uniformly bounded with
respect to $n$ and $\hslash$.
\end{proof}

\subsection{Dispersive estimates}
Let us now discuss the dispersive properties of the Yukawa models, both
quantum and classical. Recall that we are supposing that \cref{hyp:1} is
satisfied. We start with a first elementary dispersive estimate.
\begin{lemma}\label{lem:controldecay}
  For every $\xi \in \mathscr{C}_0^\infty(\mathbb{R}\setminus\{0\})$, there
  exists $C>0$ such that, for all $t \in \mathbb{R}$,
  \begin{gather}
    \label{eq:8}
    \left\|\langle x\rangle^{-1-\nu}  \left\langle \xi_t,\lambda_{x} \right\rangle\right\|_{L^\infty}\leq \frac{C}{\langle t\rangle^{1+\nu}}\,,\\
    \label{eq:9}
    \|(-\Delta + V+1)^{-1/2}\, \Im \left\langle \xi_t,\lambda_x \right\rangle\,(-\Delta + V+1)^{-1/2}\|_{\mathscr{B}(L^2(\mathbb{R}^{d}))}\leq \frac{C}{\langle t\rangle^{1+\nu}}\,,
  \end{gather}
  where $\xi_t = e^{-it\omega} \xi$.
\end{lemma}
\begin{proof}
  Observe that
  \begin{multline}
    \|(-\Delta + V+1)^{-1/2}\, \Im \left\langle \xi_t,\lambda_x \right\rangle\,(-\Delta +
    V+1)^{-1/2}\|_{\mathscr{B}(L^2(\mathbb{R}^{d}))} \\ \leq\|\langle \,\cdot\,\rangle^{-1-\nu}
    \left\langle \xi_t,\lambda_{(\cdot)} \right\rangle \|_{L^\infty} \,\|(-\Delta + V+1)^{-1/2}\langle
    x\rangle^{\frac{1+\nu}{2}}\|^2_{\mathscr{B}(L^2(\mathbb{R}^{d}))} \,.
  \end{multline}
  The last term on the right-hand side is bounded by a constant, thanks to
  the assumption \eqref{hyp:pot} on $V$.  Moreover, $\xi \in
  \mathscr{C}_0^{\infty}(\mathbb{R}^d \setminus \{0\})$ and thus the
  stationary point $\{k: \nabla \omega(k) = 0\}=\{0\}$ does not belong to the
  support of $\xi$. Hence, the non-stationary phase method yields \citep[see,
  \emph{e.g.},][Theorem XI.14]{reed1979III}:
  \begin{align*}
    \langle x\rangle^{-1-\nu} \bigl| \left\langle \xi_t,\lambda_{x} \right\rangle\bigr| &= \langle x\rangle^{-1-\nu} \left|\int_{\mathbb{R}^d} \mathrm{d} k \, e^{i t\omega(k)} \,\overline{\xi}(k) \,\lambda_x(k) \right| \leq \frac{C}{\langle t \rangle^{1+\nu} }\;.
  \end{align*}

\end{proof}

We prove below the main $\hslash$-uniform decay estimate for the Yukawa
model. Recall that the operator $S$ is defined according to \eqref{eq:13}.
\begin{proposition}\label{prop:timedecay}
  For any $\xi \in \mathscr{C}_0^\infty(\mathbb{R}^d \setminus \{0\})$ there
  exists $c>0$ such that, for all $t\in\mathbb{R}$,
  \[
    \left\| S^{-1/2} \,\mathrm{d} \Gamma_{\hslash}^{(1)}\bigl(\Im
      \left\langle \xi_t ,\lambda_{(\cdot)} \right\rangle\bigr)\, S^{-1/2}
    \right\|_{\mathscr{B}(\mathscr{H})} \leq \frac{c}{\langle t
      \rangle^{1+\nu}}\,,
  \]
  uniformly with respect to $\hslash\in(0,1)$.
\end{proposition}
\begin{proof}
  Consider $\Phi^{(n)}, \Psi^{(n)} \in \mathscr{H}^{(n)}$, then using
  symmetry one obtains
  \begin{align*}
    \Bigl|\bigl\langle &\Phi^{(n)},S^{-1/2} \,\mathrm{d} \Gamma^{(1)}_{\hslash}(\Im  \langle \xi_t,\lambda_{(\cdot)} \rangle) \, S^{-1/2}  \Psi^{(n)} \bigr\rangle\Bigr| \\ &= \left| \hslash n  \bigl\langle S^{-1/2} \Phi^{(n)},\Im  \langle \xi_t,\lambda_{x_1} \rangle S^{-1/2}\Psi^{(n)} \bigr\rangle \right|  \\
    &\leq \Bigl\|\, \frac{\Im \left\langle \xi_t,\lambda_{x_1} \right\rangle}{\langle x_1\rangle^{1+\nu}} \,\Bigr\|_{\mathscr{B}(L^2(\mathbb{R}^{d}))}  \Bigl\|\sqrt{(n\hslash)\langle x_1\rangle^{1+\nu} }S^{-1/2}\Phi^{(n)}\Bigr\|_{\mathscr{H}^{(n)}} \Bigl\|\sqrt{(n\hslash)\langle x_1\rangle^{1+\nu} }S^{-1/2} \Psi^{(n)}\Bigr\|_{\mathscr{H}^{(n)}}\\
    & \leq  \frac{C}{\langle t\rangle^{1+\nu}} \|\sqrt{(n\hslash)\langle x_1\rangle^{1+\nu} }S^{-1/2}\Phi^{(n)}\|_{\mathscr{H}^{(n)}} \|\sqrt{(n\hslash)\langle x_1\rangle^{1+\nu} }S^{-1/2} \Psi^{(n)}\|_{\mathscr{H}^{(n)}}\,,
  \end{align*}
  where in the last inequality \cref{lem:controldecay} is used. Using
  symmetry again and assumption \eqref{hyp:pot},
  \begin{align*}
    \|\sqrt{(n\hslash)\langle x_1\rangle^{1+\nu} }S^{-1/2} \Psi^{(n)}\|_{\mathscr{H}^{(n)}}^2&=
    \bigl\langle \Psi^{(n)},S^{-1/2} \,\mathrm{d} \Gamma^{(1)}_{\hslash}(\langle x\rangle^{1+\nu}) \, S^{-1/2}\Psi^{(n)}\bigr\rangle\\
    &\leq C' \|\ \Psi^{(n)}\|_{\mathscr{H}^{(n)}}^2\; .
  \end{align*}
  An analogous inequality holds for $\Phi^{(n)}$ as well, and therefore
  \begin{equation*}
    \Bigl|\bigl\langle \Phi^{(n)},S^{-1/2} \,\mathrm{d} \Gamma^{(1)}_{\hslash}(\Im  \langle \xi_t,\lambda_{(\cdot)} \rangle) \, S^{-1/2}  \Psi^{(n)} \bigr\rangle\Bigr|\leq\frac{c}{\langle t \rangle^{1+\nu}} \|\ \Phi^{(n)}\|_{\mathscr{H}^{(n)}}\|\ \Psi^{(n)}\|_{\mathscr{H}^{(n)}}\; .
  \end{equation*}
\end{proof}

\section{Scattering theory of the Schr\"odinger-Klein-Gordon equation}
\label{sec:3}

In the sequel, we discuss the classical scattering theory of the
Schrödinger-Klein-Gordon system.  Since we assume that the Schrödinger
particle is confined, only Klein-Gordon waves can be asymptotically
free. Therefore, in our context the diffusion scheme differs from the
traditional translation invariant case studied for instance in
\cite{MR1275411,MR2037763}.

The S-KG equation is an infinite dimensional classical Hamiltonian system
described by its Hamiltonian function
\begin{equation}
  \label{eq.skg-eng}
  \mathscr{E}(u,z) = \left\langle u,(-\Delta + V)u \right\rangle_{\mathfrak{H}} + \left\langle z,\omega z \right\rangle_{\mathfrak{H}} + \int_{\mathbb{R}^{2d}}^{}\bigl(\lambda_x(k)\bar{z}(k)+\bar{\lambda}_x(k)z(k) \bigr) \lvert u(x)  \rvert_{}^2 \mathrm{d}x \mathrm{d}k\; ,
\end{equation}
defined on the energy space
\begin{equation}
  \label{eq.enspa}
  \mathscr{X}= \mathscr{D}(\sqrt{-\Delta + V}) \oplus \mathscr{D}(\sqrt{\omega})\,\subset \mathscr{Z}=\mathfrak{H}\oplus \mathfrak{H}\,,
\end{equation}
with $\mathfrak{H}=L^2(\mathbb R^d)$. The energy of the non-interacting
system is given by
\begin{align*}
  \mathscr{E}_0(u,z) = \left\langle u,(-\Delta + V)u \right\rangle_{\mathfrak{H}}+ \left\langle z,\omega z \right\rangle_{\mathfrak{H}}\,.
\end{align*}
Then, the following rough inequalities compare $\mathscr{E}$ and
$\mathscr{E}_0$.  Indeed, there exist $c,\alpha>0$ such that for all
$(u,z)\in\mathscr{X}$,
\begin{align}
  \label{eq.estenrg1}
  & \big|\mathscr{E}(u,z) -\mathscr{E}_0(u,z) \big| \leq c \,( \mathscr{E}_0(u,z) +\|u\|_{L^2}^4)\,,\\
  \label{eq.estenrg2}
  & 0\leq \mathscr{E}_0(u,z) \leq c \,( \mathscr{E}(u,z) +\alpha \|u\|_{L^2}^4)\,.
\end{align}
With these notations, the {Schrödinger-Klein-Gordon} equation is a system of
PDE given by:
\begin{equation}\label{eq:skg}
  \begin{cases}
    i \partial_t u = (-\Delta + V) \,u +  \varphi_\chi\, u \\[2mm]
    i \partial_t z = \omega z + \omega^{-1/2}\chi \,\widehat{\lvert u \rvert^2}
  \end{cases}
\end{equation}
where
\begin{equation*}
  \varphi_\chi(x)=\int_{\mathbb{R}^d}^{}\frac{1}{\sqrt{\omega(k)}}\bigl( e^{ik\cdot x}\chi(k)\bar{z}(k)+e^{-ik\cdot x}\bar{\chi}(k)z(k)\bigr)  \mathrm{d}k
\end{equation*}
is the smeared Klein-Gordon field and $\widehat{\lvert u \rvert^2}$ is the Fourier
transform of $\lvert u \rvert^2\in L^1(\mathbb R^d)$. In particular, the S-KG equation is
globally well-posed over $\mathscr{X}$ with mass $\lVert u \rVert_{L^2(\mathbb R^{d})}$
and energy $\mathscr{E}(u,z)$ as conserved quantities \citep[see,
\emph{e.g.},][]{bachelot1984aihpnl,MR380160}. In the following, we denote by
$\Phi_t:\mathscr{X}\to\mathscr{X}$ the solution flow associated to the S-KG
equation.

In order to discuss the long-time asymptotics for the Klein-Gordon field, it
is convenient to rewrite the solution $z(t)$ using the Duhamel's integral
formula:
\begin{equation}\label{eq:duhamelfield}
  z(t) = e^{-it\omega} z_0 - i \int_0^t e^{-i(t-\tau)\omega} \omega^{-1/2}\chi  \,\widehat{\lvert u(\tau)  \rvert^2} \,\mathrm{d} \tau\,.
\end{equation}
The corresponding wave operator is then defined as follows.
\begin{definition}
  \label{def.waveop}
  The classical forward and backward wave operators of the S-KG equation are
  defined as the maps:
  \begin{eqnarray*}
    \Lambda^{\pm} : \;\mathscr{X} &\longrightarrow & \mathfrak{H}\\
    (u_0,z_0) &\longmapsto & z^{\pm}:=\wlim_{t \rightarrow \pm \infty} \; e^{it\omega} z(t)\,,
  \end{eqnarray*}
  where $(u(\cdot,t),z(\,\cdot\,,t))=\Phi_t(u_0,z_0)$ is the unique solution
  of the S-KG equation \eqref{eq:skg} satisfying the initial condition
  $(u_0,z_0)$ at time $t=0$ and the limit in the right hand side is with
  respect to the weak $L^2(\mathbb{R}^d)$-topology.
\end{definition}

The above classical wave operators exist according to the proposition below.
\begin{proposition}
  \label{prop:1}
  The wave operators $\Lambda^{\pm}$ of the S-KG equation are well defined.
\end{proposition}
\begin{proof}
  Consider a smooth function $\varphi \in \mathscr{C}_0^{\infty}(\mathbb{R}^d
  \smallsetminus \{ 0 \})$, then according to \eqref{eq:duhamelfield},
  \begin{equation}
    \label{eq:cook}
    \begin{aligned}
      \langle \varphi,e^{it\omega} z(t) \rangle_{\mathfrak{H}} &= \langle \varphi,z_0 \rangle_{\mathfrak{H}} -i \bigl\langle \varphi, \int_0^t e^{i\tau\omega} \omega^{-1/2}\chi  \,\widehat{\lvert u(\tau)  \rvert^2}\mathrm{d} \tau \bigr\rangle_{\mathfrak{H}}  \\
      &=\langle \varphi,z_0\rangle_{\mathfrak{H}} -i \int_0^t \bigl\langle \varphi, e^{i\tau\omega} \omega^{-1/2}\chi  \,\widehat{\lvert u(\tau)  \rvert^2}\bigr\rangle_{\mathfrak{H}}\mathrm{d} \tau\; .
    \end{aligned}
  \end{equation}
  Using the dispersive estimate in \cref{lem:controldecay}, one obtains
  \begin{eqnarray*}
    \bigl\lvert\bigl\langle \varphi, e^{i\tau\omega} \omega^{-1/2}\chi \,\widehat{\lvert u(\tau)  \rvert^2}\bigr\rangle_{\mathfrak{H}}\bigr\rvert &\leq& \bigl\|\langle \,\cdot\, \rangle^{\frac{1+\nu}{2}} u(\tau)\bigr\|_{L^2}^2 \; \bigl\| \langle \,\cdot\, \rangle^{-1-\nu}   \langle e^{-i\tau\omega} \varphi , \lambda_{(\cdot )}\rangle_{\mathfrak{H}}
    \bigr\|_{L^\infty} \\ &\leq&  \frac{C}{\langle \tau\rangle^{1+\nu}}\; \bigl\|\langle \,\cdot\, \rangle^{\frac{1+\nu}{2}} u(\tau)\bigr\|_{\mathfrak{H}}^2\; .
  \end{eqnarray*}
  Now, the bound \eqref{eq.estenrg2} yields
  \begin{eqnarray*}
    \bigl\|\langle\,\cdot\,\rangle^{\frac{1+\nu}{2}} u(\tau)\|_{\mathfrak{H}}^2 &\leq&  \bigl\langle u(\tau),(-\Delta + V)u(\tau) \bigr\rangle_{\mathfrak{H}}   \\ &\leq& \mathscr{E}_0(u(\tau), z(\tau))
    \\ &\leq& c \, \Bigl(\mathscr{E}(u(\tau), z(\tau))+ \alpha \|u(\tau)\|^4_\mathfrak{H} \Bigr)\,.
  \end{eqnarray*}
  Thus, the energy and mass conservation guarantees that the integral in
  \eqref{eq:cook} converges as $t\to \pm\infty$. Such convergence is lifted to all $\varphi\in
  L^2(\mathbb{R}^d)$ by a density argument, using also the fact that the norm
  $\|z(t)\|_{L^2(\mathbb{R}^d)}$ is uniformly bounded in time thanks to the energy
  estimate \eqref{eq.estenrg2}.
\end{proof}

Let us make a couple of remarks.
\begin{itemize}
\item [(i)] The classical wave operators $ \Lambda^{\pm}$ have the following
  representation, intended as a distribution equality in $\mathscr{D}'(\mathbb{R}^d\smallsetminus
  \{0\})$:
  \begin{equation}
    \label{eq:waverepres}
    \Lambda^{\pm}(u_0,z_0) = z_0 - i \int_0^{\pm\infty} e^{i\tau\omega} \omega^{-1/2}\chi \,\widehat{\lvert u(\tau)  \rvert^2} \mathrm{d} \tau\,.
  \end{equation}
\item [(ii)] For all $\zeta\in \mathfrak{H}$,
  \begin{equation*}
    \lim_{t\to \pm\infty} \langle e^{-it\omega}\zeta  , z(t)-e^{-it\omega}z^{\pm} \rangle_{\mathfrak{H}}=0\; .
  \end{equation*}
  In other words, the nonlinear evolution of the field $z_0$ can be
  approximated for long times, in a weak sense, by the free evolution of
  $z^{\pm}=\Lambda^{\pm}(u_0,z_0)$.
\end{itemize}
Note also that despite the name, $\Lambda^\pm$ are not linear maps.

\begin{corollary}
  \label{cor:2}
  \begin{equation*}
    \ran \Lambda^{\pm}\subseteq \mathscr{D}(\sqrt{\omega})\; .
  \end{equation*}
\end{corollary}
\begin{proof}
  The solution $z(t)\in \mathscr{D}(\sqrt{\omega})$ for all $t\in
  \mathbb{R}$. In addition, by \eqref{eq.estenrg2} and the conservation of
  energy, $\lVert \sqrt{\omega} z(t) \rVert_{L^2}$ is uniformly bounded in
  time. Therefore, the convergence in \cref{prop:1} holds when testing with
  $\varphi\in \omega^{-1/2}L^2(\mathbb{R}^d)$ as well.
\end{proof}

As for the quantum theory and asymptotic vacuum states \eqref{eq:19}, at the
classical level there is a notion of asymptotic radiationless states. These
are the phase space points in the kernel of the classical wave operator, {\it
  i.e.}
\begin{equation}
  \label{eq:5}
  \mathscr{K}^{\pm}_0=\bigl\{(u,z)\in \mathscr{X}\;\vert\; \Lambda^{\pm}(u,z)=0\bigr\}\; .
\end{equation}
In the next sections, we will see how to relate the notions of classical wave
operators and the space of asymptotic radiationless states to the
corresponding ones for the quantum Yukawa model.

\section{Semiclassical limit of meson fields}
\label{sec:4}

In this section we study the behavior of the asymptotic fields of the Yukawa
theory, as $\hslash\to 0$.  As already recalled in \cref{subsec:wig}, at any finite
time $t\in \mathbb{R}$ the semiclassical behavior of the family $\{\Psi_\hslash(t)= e^{-i
  \frac{t}{\hslash} H_\hslash} \Psi_\hslash\}_{\hslash\in(0,1)}$ can be characterized by their
semiclassical measures at a fixed given time (usually $t=0$), pushed forward
by the nonlinear flow associated to the classical S-KG equations
\eqref{eq:skg} \citep[see][for a detailed analysis]{ammari2014jsp}. Combining
the techniques used for finite times with the dispersive properties of both
quantum and classical Yukawa models, one extends the analysis to infinite
times $t\to \pm \infty$.

\subsection{Meson fields}
For convenience, recall the operator $S$ already introduced in
\cref{subsec:est-enr}:
\begin{equation}
  \label{opS}
  S:= H_\hslash^0 + N_1^2 + 1\,.
\end{equation}
This operator emerges naturally, for this particular model, as the right tool
that encodes the regularity properties needed to characterize explicitly the
semiclassical limit of asymptotic fields.  To this extent, we will make
extensive use of the following assumption.
\begin{hyp}
  Given a family $\{\Psi_{\hslash}\}_{\hslash \in (0,1)} \subseteq
  \mathscr{H}$ of normalized vectors, there exists $C>0$ such that, uniformly
  with respect to $\hslash\in(0,1)$,
  \begin{equation}\label{hyp:states}\tag{A4}
    \left\langle \Psi_{\hslash}, S \,\Psi_{\hslash} \right\rangle \leq C.
  \end{equation}
\end{hyp}
The integral formula below is a consequence of such assumption.

\begin{proposition}\label{prop:weylrepres}
  Let $\{\Psi_{\hslash}\}_{\hslash \in (0,1)}$ be a family of vectors
  satisfying \eqref{hyp:states}. Denoting $\Psi_{\hslash}(t) =
  e^{-i\frac{t}{\hslash}H_{\hslash}} \Psi_{\hslash}$, we have that for all
  $\xi \in \mathscr{C}_0^\infty(\mathbb{R}^d \setminus \{0\})$,
  \begin{equation}\label{eq:weylrepres}
    \left\langle \Psi_{\hslash},W^{\pm}_{\hslash}(\xi) \Psi_{\hslash} \right\rangle = \left\langle \Psi_{\hslash},W_{\hslash}(\xi) \Psi_{\hslash} \right\rangle +\sqrt{2}i \int_0^{\pm\infty}  \left\langle \Psi_{\hslash}(\tau), \mathrm{d} \Gamma_{\hslash}^{(1)}\bigl( \Im  \langle \xi_\tau,\lambda_{(\cdot)} \rangle_{\mathfrak{H}}\bigr)\Psi_{\hslash}(\tau) \right\rangle \mathrm{d} \tau\; .
  \end{equation}
\end{proposition}
\begin{proof}
  The operator $e^{i\frac{t}{\hslash}H_{\hslash}} W_{\hslash}(\xi_t)
  \,e^{-i\frac{t}{\hslash}H_{\hslash}} $ is weakly differentiable on
  $\mathscr{D}(\sqrt{S})\subset \mathscr{Q}(H_\hslash)$. This is a
  consequence of \cref{lem:formcontrol,lem:controltime}, together with the
  following properties of Weyl operators:
  \begin{equation}
    W_{\hslash}(\xi_t) = e^{-i\frac{t}{\hslash} H_\hslash^0}\, W_{\hslash}(\xi) \,e^{i\frac{t}{\hslash} H_\hslash^0} \;,\quad W_{\hslash}(\xi_t) \mathscr{D}(\sqrt{S})\subset \mathscr{D}(\sqrt{S})\; ,
  \end{equation}
  for any $\xi\in \mathscr{C}_0^\infty(\mathbb{R}^d \setminus \{0\})$. The
  time derivative yields
  \begin{equation}
    \frac{\mathrm{d}}{ \mathrm{d} t}\, e^{i\frac{t}{\hslash}H_{\hslash}} W_{\hslash}(\xi_t) \,e^{-i\frac{t}{\hslash}H_{\hslash}} =\sqrt{2} i \, e^{i\frac{t}{\hslash} H_{\hslash}}\, \mathrm{d} \Gamma_{\hslash}^{(1)}\bigl(\Im  \left\langle \xi_t, \lambda_{(\cdot)} \right\rangle_{\mathfrak{H}} \bigr)W_{\hslash}(\xi_t) \, e^{-i\frac{t}{\hslash} H_{\hslash}}\,.
  \end{equation}
  By the fundamental theorem of calculus, we obtain
  \begin{equation}
    \label{eq.pr.int.f}
    \begin{split}
      \bigl\langle \Psi_{\hslash},e^{i\frac{t}{\hslash}H_{\hslash}} W_{\hslash}(\xi_t) \,e^{-i\frac{t}{\hslash}H_{\hslash}} \Psi_{\hslash}\bigr\rangle&= \bigl\langle \Psi_{\hslash}, W_{\hslash}(\xi)  \Psi_{\hslash} \bigr\rangle \\
      &+ \sqrt{2} i\int_0^{t} \bigl\langle \Psi_{\hslash}(\tau), \mathrm{d} \Gamma_{\hslash}^{(1)}\bigl( \Im  \langle \xi_\tau,\lambda_{(\cdot)} \rangle_{\mathfrak{H}} \bigr)W_{\hslash}(\xi_\tau) \,\Psi_{\hslash}(\tau) \bigr\rangle \,\mathrm{d}\tau\; .
    \end{split}
  \end{equation}
  Using \cref{lem:controltime} and the canonical commutation relations
  \citep[see, \emph{e.g.},][Lemma 2.5]{derezinski1999rmp}, it then follows
  that there exists $c>0$ such that for all $\tau\in\mathbb{R}$,
$$
\langle W_{\hslash}(\xi_\tau) \Psi_{\hslash}(\tau), S W_{\hslash}(\xi_\tau)
\,\Psi_{\hslash}(\tau)\rangle \leq c\,.
$$
Hence we can use \cref{prop:timedecay}, to find a (possibly
different) $c>0$ such that
\begin{equation}
  \label{eq.estdecapr}
  \bigl\lvert\langle \Psi_{\hslash}(\tau), \mathrm{d} \Gamma_{\hslash}^{(1)}\bigl( \Im \langle \xi_\tau,\lambda_{(\cdot)} \rangle_{\mathfrak{H}} \bigr)W_{\hslash}(\xi_\tau) \,\Psi_{\hslash}(\tau) \rangle\bigr\rvert \leq c\; \langle \tau \rangle^{-1-\nu}\,.
\end{equation}
This ensures integrability on the whole positive (negative) real line, and
therefore we can take the limit $t \rightarrow \pm \infty$ in both sides of
\eqref{eq.pr.int.f}.
\end{proof}

Thanks to the above proposition and using the techniques of infinite
dimensional semiclassical analysis developed in \citep{ammari2008ahp}, we
prove below a semiclassical characterization for asymptotic Weyl operators.
We recommend the reading of paragraph \cref{subsec:wig} before going through
this part.

\begin{thm}
  \label{thm:1}
  Let $\{\Psi_{\hslash}\}_{\hslash \in (0,1)} \subseteq \mathscr{H}$ be a
  family of normalized vectors satisfying \eqref{hyp:states} and assume that
  $\mathscr{M}\bigl(\Psi_{\hslash_n}; n\in \mathbb{N}\bigr)=\{ \mu\}$ for a
  sequence $\hslash_n\to 0$. Then the Wigner measure $\mu$ is concentrated on
  $\mathscr{X}=\mathscr{D}(\sqrt{-\Delta+V})\oplus
  \mathscr{D}(\sqrt{\omega})$, and for all $\xi\in\mathfrak{H}$,
  \begin{align}
    \lim_{n\to \infty} \bigl\langle \Psi_{\hslash_n},W^\pm_{\hslash_n}(\xi) \Psi_{\hslash_n} \bigr\rangle &= \int_{\mathscr{X}} e^{i\sqrt{2} \Re  \left\langle \xi,\Lambda^\pm (u,z) \right\rangle_{\mathfrak{H}}}\,\mathrm{d} \mu(u,z)\;, \label{eq:weyllimit2}
  \end{align}
  where $\Lambda^\pm $ are the classical wave operators of \cref{def.waveop}.

  Additionally, the following integral representation holds true for all
  $\xi\in \mathscr{C}_0^\infty(\mathbb{R}^d \setminus \{0\})$:
  \begin{equation}
    \label{eq:weyllimit1}
    \begin{aligned}
      \lim_{n\to \infty} \bigl\langle \Psi_{\hslash_n},W^\pm_{\hslash_n}(\xi) \Psi_{\hslash} \bigr\rangle &= \int_{\mathscr{X}} e^{i\sqrt{2}  \Re  \langle \xi,z \rangle_{\mathfrak{H}}}\mathrm{d} \mu(u,z)  \\
      &+  \sqrt{2}i \int_0^{\pm \infty} \int_{\mathscr{X} } e^{i\sqrt{2} \Re \langle \xi_\tau,z \rangle_{\mathfrak{H}}} \Im \bigl\langle \xi_\tau  ,  \omega^{-1/2}\chi \,\widehat{\lvert u  \rvert^2}\bigr\rangle_{\mathfrak{H}}\mathrm{d} \mu_\tau(u,z)\mathrm{d} \tau\,,
    \end{aligned}
  \end{equation}
  where $\mu_\tau= (\Phi_\tau) \, _{*}\,\mu$ is the measure pushed forward by
  the nonlinear flow $\Phi_\tau : \mathscr{X}\rightarrow \mathscr{X}$ solving
  the S-KG equation \eqref{eq:skg}.
\end{thm}
\begin{proof}
  Firstly, let us fix $\xi\in \mathscr{C}_0^\infty(\mathbb{R}^d \setminus
  \{0\})$.  By \citep[Theorem 1.1]{ammari2014jsp}, it follows that
  $\Psi_{\hslash_n}\rightharpoondown \mu$ if and only if for all $t\in
  \mathbb{R}$, $\Psi_{\hslash_n}(t)\rightharpoondown \mu_t= (\Phi_t) \,
  _{*}\,\mu$. In addition, thanks to \eqref{hyp:states}, it also follows that
  $\mu$ is concentrated on $\mathscr{X}$. Now, it is sufficient to take the
  limit $n\to \infty$ on both sides of \eqref{eq:weylrepres} (with
  $\hslash_n$ in place of $\hslash$) to obtain \eqref{eq:weyllimit1}. The
  strategy for the proof is analogous to the one used to prove the
  aforementioned \citep[Theorem 1.1]{ammari2014jsp}, let us however briefly
  comment on how to deal with the term containing the integral over all
  positive times $\tau$. The crucial point is that the bound for the
  integrand given in \eqref{eq.estdecapr} is uniform with respect to
  $\hslash$. Therefore, Lebesgue's dominated convergence theorem shall be
  applied, exchanging the integral with the limit $n\to \infty$, and then the
  convergence follows from the same arguments as in the proofs of Lemma 3.14
  and Proposition 4.10 in \citep{ammari2014jsp}.

  It remains to prove \eqref{eq:weyllimit2} for all $\xi\in
  \mathfrak{H}$. However, for the moment keep
  $\xi\in\mathscr{C}_0^\infty(\mathbb{R}^d \setminus \{0\})$ as above. Remark
  that, by \cref{def.waveop},
  \begin{equation}
    \label{eq:intermediatelimit}
    \begin{split}
      \int_{\mathscr{X} } e^{i\sqrt{2} \Re  \langle \xi,\Lambda^\pm(u,z) \rangle_{\mathfrak{H}}}&\mathrm{d} \mu(u,z) = \lim_{t \rightarrow \pm \infty} \int_{\mathscr{X}} e^{i\sqrt{2} \Re  \langle \xi_t,z(t)\rangle_{\mathfrak{H}}}\mathrm{d} \mu(u,z) \\
      =\int_{\mathscr{X}} &e^{i\sqrt{2} \Re  \langle \xi,z\rangle}\mathrm{d} \mu(u,z)  +\lim_{t\to \pm \infty} \int_0^t  \int_{\mathscr{X}}\frac{\mathrm{d}}{\mathrm{d} \tau} \biggl(e^{i\sqrt{2} \Re  \langle \xi_\tau , z(\tau) \rangle_{\mathfrak{H}}}\biggr)\mathrm{d} \mu(u,z)\mathrm{d} \tau\; .
    \end{split}
  \end{equation}
  Now, S-KG equation \eqref{eq:skg} yields precisely
  \begin{align*}
    \frac{\mathrm{d}}{\mathrm{d} \tau} e^{i\sqrt{2} \Re  \langle \xi_\tau , z(\tau)\rangle_{\mathfrak{H}}}  &= \sqrt{2}i e^{i\sqrt{2} \Re \langle \xi_\tau,z(\tau) \rangle_{\mathfrak{H}}} \Im \bigl\langle \xi_\tau  ,  \omega^{-1/2}\chi \,\widehat{\lvert u(\tau)  \rvert^2}\bigr\rangle_{\mathfrak{H}}\; .
  \end{align*}
  Hence \eqref{eq:weyllimit2} is proved for all
  $\xi\in\mathscr{C}_0^\infty(\mathbb{R}^d \setminus \{0\})$.  To extend the
  proof to all $\xi\in \mathfrak{H}$, a density argument is used, observing
  that by \cref{lem:formcontrol} and \citep[][Lemma 3.1]{ammari2008ahp} it
  follows that, uniformly with respect to $\hslash$, there exists $c>0$ such
  that
  \begin{equation}
    \label{eq.ws}
    \left\|\big(W_\hslash^\pm(\xi)-W_\hslash^\pm(\eta)\big) S^{-1/2}\right\|_{\mathscr{B}(\mathscr{H})} \leq c \; \|\xi-\eta\|_{\mathfrak{H}}\;.
  \end{equation}
\end{proof}
\begin{remark}
  The asymptotic Weyl operators have the following intertwining property
  \citep[see][Theorem 5.1]{derezinski1999rmp}:
$$
e^{i\frac{t}{\hslash}H_{\hslash}} \;W_\hslash^\pm(\xi) \;
e^{-i\frac{t}{\hslash}H_{\hslash}} = W_\hslash^\pm(\xi_{-t})\,.
$$
\cref{thm:1} yields that the intertwining property is also inherited by the
classical wave operator:
$$
\Lambda^{\pm}(\Phi_t(u,z))= e^{it \omega} \Lambda^{\pm}(u,z)\,.
$$
Of course, such relation can also be proved using directly \cref{def.waveop}
and \cref{prop:1}.
\end{remark}

The previous results connect, via the semiclassical limit, the quantum and
classical scattering theories of the Yukawa model and the S-KG system,
respectively. For the moment, we have shown that the time asymptotic limit of
the Fourier-Wigner transforms of quantum states converge to the
characteristic functions of probability measures, that are none other than the
initial Wigner measures at time $t=0$ pushed forward by the classical wave
operators $\Lambda^\pm$. Now, we prove that quantum asymptotic meson fields converge
to classical asymptotic Klein-Gordon fields.

Recall that starting from the asymptotic Weyl operators $W^{\pm}_{\hslash}$
it is possible to define the asymptotic fields $\phi_{\hslash}^{\pm}$ as
their self-adjoint generators:
\begin{equation*}
  W^{\pm}_{\hslash}(\eta)=e^{\frac{i}{\sqrt{2}}\phi^{\pm}_{\hslash}(\eta)}\; .
\end{equation*}
These operators have a common dense core, given by the square root of $S$:
for all $\xi\in\mathfrak{H}$, $\mathscr{D}(\sqrt{S})\subset \mathscr{D}(\phi^{\pm}_{\hslash}(\xi))$.

\begin{proposition}\label{prop:segalfield}
  Let $\{\Psi_{\hslash}\}_{\hslash \in (0,1)} \subseteq \mathscr{H}$ be a
  family of normalized vectors satisfying \eqref{hyp:states} and assume that
  $\Psi_{\hslash_n}\rightharpoondown \mu$ for a sequence $\hslash_n\to
  0$. Then for all $\xi\in\mathfrak{H}$,
  \begin{align}
    \label{eq.asphi}
    \lim_{n\to \infty} \langle \Psi_{\hslash_n},\phi^{\pm}_{\hslash_n}(\xi)\Psi_{\hslash_n} \rangle = 2\int_{\mathscr{X}}  \Re \langle \xi, \Lambda^\pm(u,z) \rangle_{\mathfrak{H}} \,\mathrm{d} \mu(u,z)\; .
  \end{align}
\end{proposition}
\begin{proof}
  The function $t \mapsto e^{i\frac{t}{\hslash} H_{\hslash}}\,\phi_{\hslash}
  (\xi_t)\, e^{-i\frac{t}{\hslash}H_{\hslash}}$ is differentiable as a
  quadratic form on $\mathscr{D}(\sqrt{S})$, and
  \[\frac{\mathrm{d} }{\mathrm{d} t }e^{i\frac{t}{\hslash}
      H_{\hslash}}\,\phi_{\hslash} (\xi_{t})\, e^{-i\frac{t}{\hslash}
      H_{\hslash}} = 2 \,e^{i\frac{t}{\hslash} H_{\hslash}}\, \mathrm{d}
    \Gamma_{\hslash}^{(1)}(\Im \langle \xi_t,\lambda_x
    \rangle_{\mathfrak{H}}) \, e^{-i\frac{t}{\hslash} H_{\hslash}}\,.
  \]
  Therefore, we get
  \begin{align*}
    \bigl\langle \Psi_{\hslash},e^{i\frac{t}{\hslash} H_{\hslash}}\,\phi_{\hslash} (\xi_{t})\, e^{-i\frac{t}{\hslash} H_{\hslash}}\Psi_{\hslash} \bigr\rangle = \bigl\langle \Psi_{\hslash},\phi_{\hslash}(\xi)\Psi_{\hslash} \bigr\rangle + 2\int_0^t \bigl\langle \Psi_{\hslash}(\tau), \mathrm{d} \Gamma_{\hslash}^{(1)}(\Im  \langle \xi_\tau,\lambda_x \rangle_{\mathfrak{H}})\Psi_{\hslash}(\tau) \bigr\rangle\mathrm{d} \tau\;.
  \end{align*}
  Using \cref{prop:timedecay}, the integrand in the previous expression is in
  $L^1_\tau(\mathbb{R}^{\pm})$ for all $\xi\in
  \mathscr{C}_0^{\infty}(\mathbb{R}^d\smallsetminus \{0\})$. Hence, taking
  the limit $t \rightarrow \pm \infty$ yields:
  \begin{equation}\label{eq:timesegal}
    \begin{split}
      \bigl\langle \Psi_{\hslash_n},\phi^\pm_{\hslash_n} (\xi)\Psi_{\hslash_n}\bigr\rangle =& \bigl\langle \Psi_{\hslash_n},\phi_{\hslash_n}(\xi)\Psi_{\hslash_n} \bigr\rangle \\&+ 2 \int_0^{\pm \infty} \bigl\langle \Psi_{\hslash_n}(\tau), \mathrm{d} \Gamma_{\hslash_n}^{(1)}(\Im  \langle \xi_\tau,\lambda_x \rangle_{\mathfrak{H}})\Psi_{\hslash_n}(\tau) \bigr\rangle\,\mathrm{d} \tau\;.
    \end{split}
  \end{equation}
  Using again \cite[][Theorem 1.1]{ammari2014jsp}, we observe that
  $\mathscr{M}(\Psi_{\hslash_n}(\tau) ; n\in \mathbb{N})=\{\mu_\tau =
  (\Phi_\tau)\, _{*}\,\mu\}$. Hence, by \cref{lem:controltime}, and the
  analogous of \citep[][Lemma 3.15 and Proposition 4.10]{ammari2014jsp}, the
  limit $\hslash_n \rightarrow 0$ of \eqref{eq:timesegal} yields:
  \begin{equation*}
    \lim_{n\to \infty} \bigl\langle \Psi_{\hslash_n},\phi^\pm_{\hslash_n} (\xi)\Psi_{\hslash_n} \bigr\rangle = 2 \int_{\mathscr{X}}  \Re  \langle \xi,z\rangle_\mathfrak{H}\mathrm{d} \mu(u,z) + 2\int_0^{\pm \infty} \mspace{-7mu}\int_{\mathscr{X} }  \Im \bigl\langle \xi_\tau  ,  \omega^{-1/2}\chi \,\widehat{\lvert u  \rvert^2}\bigr\rangle_{\mathfrak{H}} \mathrm{d} \mu_\tau(u,z)\mathrm{d} \tau\; .
  \end{equation*}
  Recalling the representation formula \eqref{eq:waverepres} for the
  classical wave operators, it follows that
  \begin{equation}
    \Re  \langle \xi,z \rangle_\mathfrak{H} + \int_0^{\pm\infty}  \Im \bigl\langle \xi_\tau  ,  \omega^{-1/2}\chi \,\widehat{\lvert u(\tau)  \rvert^2}\bigr\rangle_{\mathfrak{H}} \mathrm{d} \tau =  \Re  \langle \xi, \Lambda^\pm(u,z)\rangle_\mathfrak{H}\;.
  \end{equation}
  The result extends to any $\xi \in \mathfrak{H}$ by means of the uniform
  bound
  \begin{equation}
    \label{eq:27}
    \left\|\big( \phi^\pm_{\hslash} (\xi)-\phi^\pm_{\hslash} (\eta)\big) S^{-1/2}\right\|\leq c \,\|\xi-\eta\|_{\mathfrak{H}}\; .
  \end{equation}
  Let us remark that the energy estimates
  \eqref{eq.estenrg1}-\eqref{eq.estenrg2} imply the bound
  \[
    \|\Lambda^\pm(u,z)\|_{\mathfrak{H}} \leq \frac{1}{m}
    \mathscr{E}_0(u,z)\leq c\; ( \mathscr{E}(u,z)+\alpha
    \|u\|^4_\mathfrak{H})\,,
  \]
  and that by assumption \eqref{hyp:states} we have that
$$
\int_{\mathscr{X}}\mathscr{E}_0(u,z) \,{\rm d}\mu(u,z) <\infty\;.
$$
\end{proof}

\begin{corollary}\label{cor:asymptcreat}
  Under the same assumptions of \cref{prop:segalfield}, we have that for any
  $\xi \in \mathfrak{H} $,
  \begin{equation}\label{eq:creatlimit}
    \begin{split}
      \lim_{n\to \infty} \langle \Psi_{\hslash_n},a^{\pm }_{\hslash_n} (\xi)\Psi_{\hslash_n}\rangle &= \int_{\mathscr{X} }  \langle \xi,\Lambda^\pm(u,z) \rangle_\mathfrak{H} \;\mathrm{d} \mu(u,z)\;,\\
      \lim_{n\to \infty} \langle \Psi_{\hslash_n},a^{\pm,*}_{\hslash_n} (\xi)\Psi_{\hslash_n}\rangle &= \int_{\mathscr{X} }  \langle \Lambda^\pm(u,z), \xi \rangle_\mathfrak{H} \;\mathrm{d} \mu(u,z)\; .
    \end{split}
  \end{equation}
\end{corollary}

\subsection{Correlation functions and transition amplitudes}
\label{sec:class-limit-asympt}

In the quantum Yukawa model, the number of nucleons is invariant
(\emph{i.e.}, the corresponding operator $N_1$ strongly commutes with the
Hamiltonian $H_{\hslash}$). As a consequence, the asymptotic operators $W^{\pm}_{\hslash},
\phi^\pm_{\hslash}, a^{\pm,\natural }_{\hslash} $ all commute in a strong sense with $N_1$, and they
can be decomposed to a direct sum of operators on each fiber
$\mathscr{H}^{(n)}$, see \eqref{eq:1}, \eqref{eq:22} and
\eqref{eq:14}. Hence, it is natural to consider the following type of quantum
states.
\begin{hyp}
  Let $\big\{\Psi^{(n_{k})}_{\hslash_{k}}\big\}_{k\in \N}$ be a family of normalized
  vectors on $\mathscr{H}$ such that there exists $\delta>0$ such that
  \begin{equation}
    \label{hyp:secwf} \tag{A5}
    N_1 \Psi^{(n_k)}_{\hslash_k}= \hslash_k \;n_k\; \Psi^{(n_k)}_{\hslash_k}\,,\qquad \lim_{k\to\infty} \hslash_k=0\,, \qquad  \lim_{k\to\infty} \hslash_kn_k=\delta^2\,, \qquad \text{and} \qquad  \Psi^{(n_k)}_{\hslash_k}\rightharpoondown \mu.
  \end{equation}
\end{hyp}
Recall that the last statement $\Psi^{(n_k)}_{\hslash_k}\rightharpoondown \mu$ means that the sequence
$\big\{\Psi^{(n_{k})}_{\hslash_{k}}\big\}_{k\in \N}$ admits a unique Wigner measure
according to \eqref{wigner}.  On such families of vectors it is possible to
study the semiclassical behavior of the asymptotic $p$-point correlation
functions and to deduce relevant informations on asymptotic vacuum vectors,
bound states (see \cref{subsec:bds}) and ground states (see
\cref{sec:semicl-limit-ground}).  The asymptotic $p$-point correlation
functions are defined as follows.  Let $\eta\in \mathscr{S}(\mathbb{R}^d, \mathbb{R})\subset
L^2(\mathbb{R}^d)=\mathfrak{H}$, then by \eqref{eq:27}, or analogously \eqref{eq:4},
the asymptotic fields $\phi^{\pm}_{\hslash}(\cdot )$ can be seen as operator valued
distributions in momentum space. That is usually written as
\begin{equation*}
  \phi^{\pm}_{\hslash}(\eta)= \int_{\mathbb{R}^d}^{}\phi^{\pm}_{\hslash}(k)\eta(k)  \mathrm{d}k\; ,
\end{equation*}
with $\phi^{\pm}_{\hslash}(k)$ the aforementioned operator-valued distributions. Taking
the Fourier transform on all $L^2$-wavefunctions, a unitary transformation
$\mathcal{F}$ is induced on the meson field's Fock space. Using such a
unitary transformation, it is possible to define the fields in position space
as operator valued distributions:
\begin{equation*}
  \varphi^{\pm}_{\hslash}(h)= \int_{\mathbb{R}^d}^{}\varphi^{\pm}_{\hslash}(x)h(x)  \mathrm{d}x\; ,
\end{equation*}
where
\begin{equation}
  \label{eq:18}
  \varphi^{\pm}_{\hslash}(\check{\eta})= \mathcal{F}^{-1}\phi^{\pm}_{\hslash}(\eta) \mathcal{F}\; .
\end{equation}
We have here denoted by $\check{\eta}$ the inverse Fourier transform of $\eta$.

Given a vector $\Psi_{\hslash}^{(n)}\in \mathscr{H}^{(n)}$ (in the field's
momentum Fock space), the $p$-point asymptotic correlation functions for the
meson field are distributions in $\mathscr{S}'(\mathbb{R}^{dp})$ usually
defined as
\begin{equation}
  \label{eq:15}
  \langle\varphi^{\pm}_{\hslash}(x_1)\dotsm\varphi^{\pm}_{\hslash}(x_p)\rangle_{\Psi^{(n)}_{\hslash}}=\langle \mathcal{F}^{-1}\Psi^{(n)}_{\hslash}  , \varphi^{\pm}_{\hslash}(x_1)\dotsm\varphi^{\pm}_{\hslash}(x_p) \mathcal{F}^{-1}\Psi^{(n)}_{\hslash}\rangle_{}\; .
\end{equation}
We remark that in \eqref{eq:15} the signs are either all $+$ or all
$-$. Using again \eqref{eq:4}, it follows that the distribution
$\langle\varphi^{\pm}_{\hslash}(x_1)\dotsm\varphi^{\pm}_{\hslash}(x_p)\rangle_{\Psi^{(n)}_{\hslash}}$ is well defined for all
$\Psi^{(n)}_{\hslash}\in \mathscr{Q}((H_{\hslash}^{(n)})^{p})$, and it is a square integrable
function:
\begin{equation}
  \label{eq:16}
  \langle\varphi^{\pm}_{\hslash}(x_1)\dotsm\varphi^{\pm}_{\hslash}(x_p)\rangle_{\Psi^{(n)}_{\hslash}}\in L^2_{x_1,\dotsc,x_p}(\mathbb{R}^{dp})\; .
\end{equation}
We can characterize explicitly, using the tools introduced above, the leading
order (\emph{i.e.}, the $\hslash^0$ contribution) of the asymptotic correlation
functions for the meson field.

\begin{proposition}\label{prop:asymptcreat}
  Let $\big\{\Psi^{(n_k)}_{\hslash_k}\big\}_{k\in\mathbb{N}}$ be family of
  normalized vectors satisfying \eqref{hyp:secwf}. Assume that there exist
  $p\geq 1$ and $c>0$ such that:
  \begin{equation}
    \label{eq:hyppsin}
    \forall k\in\mathbb{N}, \qquad \langle \Psi^{(n_k)}_{\hslash_k}, \,(H_{\hslash_k}^{(n_k)}+b)^{p}\;\Psi^{(n_k)}_{\hslash_k}\rangle\leq c\; .
  \end{equation}
  Then the semiclassical measure $\mu$ in \eqref{hyp:secwf} is concentrated
  on the set $\{(u,z)\in \mathscr{X}\,,\, \lVert u
  \rVert_{\mathfrak{H}}^{}=\delta\}\subset \mathscr{X} $. Moreover, for all
  integers $\underline{p}\in\N$, $0<\underline{p}\leq 2p-1$ and all
  $\xi_1,\dotsc,\xi_{\underline{p}}\in \mathfrak{H}$,
  \begin{equation}
    \label{eq:26}
    \lim_{k \rightarrow \infty } \bigl\langle \Psi_{\hslash_k}^{(n_k)},\prod_{j=1}^{\underline{p}} \phi^{\pm}_{\hslash_k} (\xi_j)\Psi_{\hslash_k}^{(n_k)} \bigr\rangle=\int_{\mathscr{X}}^{}\prod_{j=1}^{\underline{p}}\Bigl(\langle \xi_j  , \Lambda^{\pm}(u,z) \rangle_{\mathfrak{H}}+ \langle \Lambda^{\pm}(u,z)  , \xi_j \rangle_{\mathfrak{H}}\Bigr)  \mathrm{d}\mu(u,z)\; ,
  \end{equation}
  and
  \begin{equation}
    \label{eq:17}
    \langle\varphi^{\pm}_{\hslash_k}(x_1)\dotsm\varphi^{\pm}_{\hslash_k}(x_{\underline{p}})\rangle_{\Psi^{(n_k)}_{\hslash_k}}= \int_{\mathscr{X}}^{}\prod_{j=1}^{\underline{p}} \Bigl(\check{\Lambda}^{\pm}(u,z)+\bar{\hat{\Lambda}}^{\pm}(u,z)\Bigr)(x_j)  \,\mathrm{d}\mu(u,z) + o_{\hslash_k}(1)\; ,
  \end{equation}
  where $o_{\hslash_k}(1)$ is converging to zero in the weak
  $L^2_{x_1,\dotsc,x_{\underline{p}}}(\mathbb{R}^{d \underline{p}})$
  topology.
\end{proposition}
\begin{proof}
  The concentration property of the measure $\mu$ follows from well-known
  semiclassical considerations
  \citep[see][]{ammari2008ahp,ammari2014jsp,ammari2016cms}. Using
  \eqref{eq:18}, it is straightforward to see that, in order to prove
  \eqref{eq:17}, it suffices to prove the convergence \eqref{eq:26}.

  Let us omit, for convenience, the explicit $k$-dependence of $n_k$ and
  $\hslash_k$, and denote
  \begin{equation*}
    \Psi(t)=e^{-i\frac{t}{\hslash}H_{\hslash}^{(n)}}\Psi_{\hslash}^{(n)}\; .
  \end{equation*}
  Let us also remark that
  \begin{eqnarray}
    \partial_t \prod_{j=1}^{\underline{p}} \phi_\hslash(\xi_{j,t}) = -\frac{i}{\hslash} \bigl[ H_\hslash^{0,(n)} , \prod_{j=1}^{\underline{p}} \phi_\hslash(\xi_{j,t})\bigr]\;.
  \end{eqnarray}
  Hence
  \begin{equation}
    \label{eq:integformpiphi}
    \begin{split}
      \bigl\langle \Psi(t),\prod_{j=1}^{\underline{p}} \phi_{\hslash}(\xi_{j,t})\Psi (t)\bigr\rangle = \bigl\langle \Psi,\prod_{j=1}^{\underline{p}} \phi_{\hslash}(\xi_{j})\Psi \bigr\rangle+\frac{i}{\hslash}
      \int_0^t \bigl\langle \Psi(\tau),\bigl[ H_{\hslash}^{I,(n)},  \prod_{j=1}^{\underline{p}} \phi_{\hslash}(\xi_{j,\tau})\bigr] \Psi(\tau)\bigr\rangle \mathrm{d}\tau  .
    \end{split}
  \end{equation}
  The commutator yields
  \begin{equation*}
    \bigl[ H_{\hslash}^{I,(n)},  \prod_{j=1}^{\underline{p}} \phi_{\hslash}(\xi_{j,\tau})\bigr]=\hslash \sum_{\ell=1}^n \bigl[  \phi_{\hslash}(\lambda_{x_\ell}) ,  \prod_{j=1}^{\underline{p}} \phi_{\hslash}(\xi_{j,\tau})\bigl]=i \hslash^2 \sum_{\ell=1}^n \sum_{j=1}^{\underline{p}} \Im (\langle\lambda_{x_\ell}, \xi_{j,\tau}\rangle) \prod_{r\neq j}  \phi_{\hslash}(\xi_{r,\tau})\;.
  \end{equation*}
  In particular, \cref{lem:controldecay,prop:enrgynumberest} and
  \eqref{eq:hyppsin} yield
  \begin{eqnarray}
    \label{eq.decay.p}
    \mathcal{A}(\tau)=    \Bigl\lvert \frac{i}{\hslash}\bigl\langle \Psi(\tau),\bigl[ H_\hslash^{I,(n)},\prod_{j=1}^{\underline{p}} \phi_\hslash(\xi_{j,\tau})\bigr]\Psi(\tau)\bigl\rangle\Bigr\rvert\leq c \langle \tau\rangle^{-1-\nu}\; .
  \end{eqnarray}
  Indeed, one has
  \begin{eqnarray*}
    \mathcal{A}(\tau)& \leq & \hslash \sum_{\ell=1}^n \sum_{j=1}^{\underline{p}} \Bigl \lvert \bigl\langle \Psi(\tau),  \Im (\langle\lambda_{x_\ell}, \xi_{j,\tau}\rangle) \prod_{r\neq j}  \phi_{\hslash}(\xi_{r,\tau})\Psi(\tau)\bigl\rangle\Bigr\rvert
    \\
    & \leq & c \, \hslash \sum_{\ell=1}^n \sum_{j=1}^{\underline{p}}  \| (H_{\hslash}^{(n)}+b)^{-\frac{p}{2}} \Im (\langle\lambda_{x_\ell}, \xi_{j,\tau}\rangle) \prod_{r\neq j}  \phi_{\hslash}(\xi_{r,\tau}) (H_{\hslash}^{(n)}+b)^{-\frac{p}{2}} \|\,.
  \end{eqnarray*}
  Moreover, using an interpolation argument, for instance Hadamard's three
  lines theorem, one deduces from Proposition \ref{prop:enrgynumberest} (iii)
  the inequality:
  \begin{align*}
    &\bigl\|(H_\hslash^{(n)}+b)^{-\frac{p}{2}} (H_\hslash^{0,(n)}+1)^{\frac{1}{2}}  (N_2+1)^{\frac{p-1}{2}} \bigl\|  \leq c\,.
  \end{align*}
  Hence by standard number estimates one gets
  \begin{align*}
    &  \| (H_{\hslash}^{(n)}+b)^{-\frac{p}{2}} \Im (\langle\lambda_{x_\ell}, \xi_{j,\tau}\rangle) \prod_{r\neq j}  \phi_{\hslash}(\xi_{r,\tau}) (H_{\hslash}^{(n)}+b)^{-\frac{p}{2}} \|   \\ & \leq  C
    \| (H_\hslash^{0,(n)}+1)^{-1/2} (N_2+1)^{-\frac{p-1}{2}} \Im (\langle\lambda_{x_\ell}, \xi_{j,\tau}\rangle) \prod_{r\neq j}  \phi_{\hslash}(\xi_{r,\tau}) (H_\hslash^{0,(n)}+1)^{-1/2} (N_2+1)^{-\frac{p-1}{2}}   \|\;,
  \end{align*}
  and finally by Lemma \ref{lem:controldecay}, one
  concludes (for possibly different constants $C>0$):
  \begin{align*}
    &  \| (H_{\hslash}^{(n)}+b)^{-\frac{p}{2}} \Im (\langle\lambda_{x_\ell}, \xi_{j,\tau}\rangle) \prod_{r\neq j}  \phi_{\hslash}(\xi_{r,\tau}) (H_{\hslash}^{(n)}+b)^{-\frac{p}{2}} \|
    \\ & \leq  C  \left\|\langle x_\ell\rangle^{-1-\nu}  \left\langle\lambda_{x_\ell} ,  \xi_{j,\tau}\right\rangle\right\|_{L^\infty}  \| (N_2+1)^{-\frac{p-1}{2}}  \prod_{r\neq j}  \phi_{\hslash}(\xi_{r,\tau}) (N_2+1)^{-\frac{p-1}{2}}   \|
    \\ &
    \leq  \frac{C}{\langle \tau\rangle^{1+\nu}}\,.
  \end{align*}

  Thus, it is possible to exchange the limits $t\to \pm\infty$ and $k\to
  \infty$ in \eqref{eq:integformpiphi} for any
  $\xi_1,\dotsc,\xi_{\underline{p}}\in
  \mathscr{C}_0^{\infty}(\mathbb{R}^d\smallsetminus \{0\})$. This leads to
  \begin{equation*}
    \begin{split}
      \lim_{k \rightarrow 0 } \bigl\langle \Psi(t),\prod_{j=1}^{\underline{p}} \phi^{\pm}_{\hslash_k} (\xi_j)\Psi(t)\bigr\rangle = 2^{\underline{p}} \int_{\mathscr{X}} \prod_{j=1}^{\underline{p}} \Re  \langle \xi_j,z \rangle_{\mathfrak{H}} \mathrm{d} \mu(u,z) +
      2^{\underline{p}} \sum_{j=1}^{\underline{p}}   \int_0^{\pm \infty} \int_{\mathscr{X} }  \Im \bigl\langle \xi_{j,\tau}  ,  \omega^{-1/2}\chi \,\widehat{\lvert u  \rvert^2}\bigr\rangle_{\mathfrak{H}}\\
      \prod_{r\neq j}^{\underline{p}} \Re  \langle \xi_{r,\tau},z \rangle_{\mathfrak{H}}\mathrm{d} \mu_\tau(u,z)\mathrm{d} \tau \,.
    \end{split}
  \end{equation*}
  Recalling the property \eqref{eq:waverepres} of $\Lambda^\pm$, we have that
  \begin{equation*}
    \prod_{j=1}^{\underline{p}} \Re\langle \xi_j, \Lambda^\pm(u,z)\rangle_{\mathfrak{H}} =   \prod_{j=1}^{\underline{p}} \Re  \langle \xi_j,z\rangle_{\mathfrak{H}} + \sum_{j=1}^{\underline{p}}\int_0^{\pm \infty}\mspace{-10mu}\Im \bigl\langle \xi_{j,\tau}  ,  \omega^{-1/2}\chi \,\widehat{\lvert u(\tau)  \rvert^2}\bigr\rangle_{\mathfrak{H}}\prod_{k\neq j}^{\underline{p}} \Re  \langle \xi_{j,\tau},z(\tau)\rangle_{\mathfrak{H}}\mathrm{d}\tau \; ,
  \end{equation*}
  where $(u(\tau),z(\tau))=\Phi_\tau(u,z)$ is the solution at time $\tau$ of
  the S-KG equations \eqref{eq:skg} with initial datum $(u,z)$. A density
  argument concludes the proof.
\end{proof}

As a consequence of the above result, one obtains the $\hslash$-limit of the
transition amplitudes.
\begin{corollary}
  \label{cor:1}
  Under the same assumptions as in \cref{prop:asymptcreat}, for all integers
  $\underline{p}\in\N$, $0<\underline{p}\leq 2p-1$ and for all
  $\xi_1,\dotsc,\xi_{\underline{p}}\in \mathfrak{H}$:
  \begin{equation*}
    \begin{split}
      \lim_{k\to \infty}\langle \Psi^{(n_k)}_{\hslash_k}  , a^{\pm,\natural}_{\hslash_k}(\xi_1)\dotsm a^{\pm,\natural}_{\hslash_k}(\xi_{\underline{p}})\Psi^{(n_k)}_{\hslash_k} \rangle_{}&=\int_{\mathscr{X}}^{}\prod_{j=1}^{\underline{p}}\langle \xi_j  , \Lambda^{\pm}(u,z) \rangle_{\mathfrak{H}}^{\natural}\,  \mathrm{d}\mu(u,z)\; ,
    \end{split}
  \end{equation*}
  where the signs are either all $+$ or all $-$, $a^{\pm,\natural}$ is either
  $a^{\pm}$ or $a^{\pm,*}$ and respectively $\langle \xi_j ,
  \Lambda^{\pm}(\cdot) \rangle_{\mathfrak{H}}^{\natural}$ is $\langle \xi_j ,
  \Lambda^{\pm}(\cdot) \rangle_{\mathfrak{H}}$ or $\overline{\langle \xi_j ,
    \Lambda^{\pm}(\cdot) \rangle}_{\mathfrak{H}}$.
\end{corollary}

\section{Semiclassical properties of asymptotic vacuum states}
\label{sec:5}

\subsection{Bound states}
\label{subsec:bds}
In this section we study the semiclassical concentration properties of
\emph{bound states} (\emph{i.e.}, states belonging to ${\rm
  Ran}\,\mathds{1}_{\mathrm{pp}}(H_{\hslash}^{(n)})$). Recall that, as proved in
\citep{derezinski1999rmp} and already mentioned in \eqref{eq:20}, these bound
states correspond exactly to the asymptotic vacuum states in
$\mathscr{K}^{\pm,(n)}_{\hslash}$ defined in \eqref{eq:19}.  Hence, we are able to
prove the following semiclassical characterization.

\begin{thm}\label{prop:measvacuum}
  Let $\big\{\Psi_{\hslash_k}^{(n_{k})}\big\}_{k\in\N}$ be a sequence of
  normalized bound states satisfying \eqref{hyp:secwf}. Assume further that
  there exists $c>0$ such that:
  \begin{equation}
    \label{eq:hyppsinbd}
    \forall k\in\mathbb{N}, \qquad \langle \Psi^{(n_k)}_{\hslash_k}, \,(H_{\hslash_k}^{(n_k)}+b)^{3/2}
    \;\Psi^{(n_k)}_{\hslash_k}\rangle\leq c\; .
  \end{equation}
  Then its Wigner measure $\mu$ concentrates on the set
  \begin{equation*}
    \mathscr{K}^{+}_0\cap  \mathscr{K}^{-}_0\cap \{ (u,z)\in \mathscr{X}\,,\, \lVert u  \rVert_{\mathfrak{H}}^{}=\delta \}\; ,
  \end{equation*}
  where we recall that the space of classical asymptotic radiationless states
  is defined by
  \begin{equation}
    \mathscr{K}^{\pm}_0 =\{(u,z) \in \mathscr{X} \,|\, \Lambda^\pm(u,z) = 0\}\; .
  \end{equation}
\end{thm}
\begin{proof}
  On one hand, by \cref{cor:1} (applied with $p=3/2$ ) for all $\eta\in
  \mathfrak{H}$ if $\Psi^{(n_{k})}_{\hslash_k}\rightharpoondown \mu$,
  \begin{equation*}
    \lim_{k\to \infty} \langle \Psi^{(n_{k})}_{\hslash_k}  , a^{\pm,*}_{\hslash_k}(\eta) a^{\pm}_{\hslash_k}(\eta)\Psi^{(n_{k})}_{\hslash_k} \rangle_{}=\int_{\mathscr{X}}^{}\bigl\lvert \langle \eta  , \Lambda^{\pm}(u,z) \rangle_{\mathfrak{H}}  \bigr\rvert_{}^2  \mathrm{d}\mu(u,z)\; .
  \end{equation*}
  On the other hand, since $\Psi^{(n_{k})}_{\hslash_k}$ is an asymptotic
  vacuum,
  \begin{align*}
    \langle \Psi^{(n_{k})}_{\hslash_k}  , a^{\pm,*}_{\hslash_k}(\eta) a^{\pm}_{\hslash_k}(\eta)\Psi^{(n_{k})}_{\hslash_k} \rangle=\lVert a^{\pm}_{\hslash_k}(\eta)\Psi^{(n_{k})}_{\hslash_k}  \rVert_{}^{2}=0\; .
  \end{align*}
  Hence, for $\mu$-a.a.\ $(u,z)\in \mathscr{X}$, and for all $\eta\in
  \mathfrak{H}$,
  \begin{equation*}
    \langle \eta  , \Lambda^{\pm}(u,z) \rangle_{\mathfrak{H}}=0\; .
  \end{equation*}
  This proves the result, since as discussed in \cref{prop:asymptcreat} the
  measure $\mu$ is concentrated on $\{ (u,z)\in \mathscr{X}\,,\, \lVert u
  \rVert_{\mathfrak{H}}^{}=\delta \}$.
\end{proof}

\begin{corollary}
  \label{cor:4}
  Let $\big\{\Psi_{\hslash_k}^{(n_{k})}\big\}_{k\in\N}$ be a sequence of
  normalized bound states satisfying \eqref{hyp:secwf} and the bound
  \eqref{eq:hyppsin} for some $p\geq 1$. Then all asymptotic correlation
  functions are purely quantum:
  \begin{equation*}
    \langle\varphi^{\pm}_{\hslash}(x_1)\dotsm \varphi^{\pm}_{\hslash}(x_p) \rangle_{\Psi^{(n_{k})}_{\hslash_k}}= o_{\hslash}(1)\; .
  \end{equation*}
\end{corollary}

\subsection{Ground states}
\label{sec:semicl-limit-ground}

We study in this subsection the semiclassical behavior of \emph{ground
  states}. It is well-known that the massive quantum Yukawa model under
consideration, with trapped particles, has ground states by a HVZ type
theorem \citep[see, \emph{e.g.},][Theorem 4.1]{derezinski1999rmp}.  Let us
denote
\begin{equation}
  \label{eq:21}
  E_{\delta}=\inf_{\substack{(u,z)\in \mathscr{X}\\\lVert u  \rVert_{\mathfrak{H}}^{}=\delta}} \mathscr{E}(u,z)\; .
\end{equation}

\begin{thm}
  \label{thm:2}
  Let $\big\{\Phi^{(n_{k})}_{\hslash_k}\big\}_{k\in \N}$ be a sequence of
  ground states of $H_{\hslash_k}^{(n_{k})}$ such that there exists
  $\delta>0$,
  \begin{equation*}
    N_1 \Phi^{(n_k)}_{\hslash_k}= \hslash_k \;n_k\; \Phi^{(n_k)}_{\hslash_k}\,,\qquad \lim_{k\to\infty} \hslash_k=0\,, \qquad  \lim_{k\to\infty} \hslash_kn_k=\delta^2\,.
  \end{equation*}
  Then any Wigner measure $\mu\in
  \mathscr{M}\bigl(\Phi^{(n_{k})}_{\hslash_k};k\in\N\bigr)$ concentrates on
  the set
  \begin{equation*}
    \{(u,z)\in \mathscr{X}\,,\, \mathscr{E}(u,z)=E_{\delta}\,,\,\text{and } \lVert u  \rVert_{\mathfrak{H}}^{}=\delta \}\; .
  \end{equation*}
  In particular, it follows that the set $\{(u,z)\in \mathscr{X}\,,\,
  \mathscr{E}(u,z)=E_{\delta}\,,\,\text{and } \lVert u
  \rVert_{\mathfrak{H}}^{}=\delta \}$ is not empty, and thus the variational
  problem \eqref{eq:21} has minimizers for all $\delta>0$.
\end{thm}
\begin{proof}
  First of all, let us remark that
  $\mathscr{M}\bigl(\Phi^{(n_{k})}_{\hslash_k}; k\in\N\bigr)\neq \varnothing
  $ since ground states satisfy \eqref{hyp:states}. In addition, thanks to
  the identity
  $$
  \langle \Phi_{\hslash_k}^{(n_{k})} , N_1\Phi_{\hslash_k}^{(n_{k})}
  \rangle_{}=\hslash_kn_k\to \delta^2\,,
 $$
 one deduces from \citep[Lemma 5.6]{ammari2014jsp} that any Wigner measure
 $\mu\in\mathscr{M}\bigl(\Phi^{(n_{k})}_{\hslash_k}; k\in\N\bigr) $
 concentrates on the set
  $$ \mathscr{X}_\delta:=\{(u,z):  \lVert u  \rVert_{\mathfrak{H}}^{}=\delta\}\,.
  $$

  Let now $E_{\hslash_k}=\langle \Phi_{\hslash_k}^{(n_{k})} ,
  H_{\hslash_k}^{(n_{k})}\Phi_{\hslash_k}^{(n_{k})} \rangle_{}$. In
  \citep[Theorem 1.2]{ammari2014jsp} it was proved that
  \begin{equation*}
    \lim_{k\to \infty}E_{\hslash_k}= E_{\delta}\; .
  \end{equation*}
  In addition, an \emph{a posteriori} argument as the one used in obtaining a
  lower bound (\citep[Lemma 5.7]{ammari2014jsp}) shows that if one takes a
  subsequence of ground states
  $\big\{\Phi_{\hslash_k}^{(n_{k})}\big\}_{k\in\N}$, which we still denote
  the same, such that $\Phi_{\hslash_k}^{(n_k)}\rightharpoondown \mu$, then
  \begin{equation*}
    \int_{\mathscr{X}_\delta}^{}\mathscr{E}(u,z)  \mathrm{d}\mu(u,z)\leq E_{\delta}\; .
  \end{equation*}
  This concludes the proof, since $\mu$ is a probability measure, and thus
  its action is that of a convex combination of energies $\mathscr{E}(u,z)$.
\end{proof}

The existence of minimizers for \eqref{eq:21} can also be proved by direct
investigation of the functional $\mathscr{E}$, as showed by the next
proposition. Beforehand, let us prove a preparatory result.

\begin{lemma}\label{lem:functionalineq}
  For any $(u,z) \in \mathscr{X}$ with $\lVert u
  \rVert_{\mathfrak{H}}^{}=\delta$, the following inequalities are satisfied:
  \begin{enumerate}
  \item\label{item:4} $ \mathscr{E}_0(u,z) - 2\delta^2
    \|\omega^{-1/2}\chi\|_{\mathfrak{H}} \|z\|_{\mathfrak{H}}\leq
    \mathscr{E}(u,z) \leq \mathscr{E}_0(u,z) + 2\delta^2
    \|\omega^{-1/2}\chi\|_{\mathfrak{H}} \|z\|_{\mathfrak{H}}\,,$
  \item\label{item:5} $$ \mathscr{E}(u,z) \geq \inf_{\substack{u \in
        \mathscr{D}(\sqrt{-\Delta + V}),\\ \|u\|_{\mathfrak{H}}= \delta }}
    \langle u ,(-\Delta + V)u \rangle_{\mathfrak{H}} - \delta^4
    \|\omega^{-1}\chi\|_{\mathfrak{H}} \,.$$
  \end{enumerate}
\end{lemma}
\begin{proof}
  \ref{item:4}: Such inequalities are a consequence of Cauchy-Schwarz
  inequality:
  \begin{align*}
    \Bigl\lvert2 \Re\bigl\langle u(x),  \langle \lambda_x(k),z(k) \rangle_{\mathfrak{H}_k} u(x) \bigr\rangle_{\mathfrak{H}_x}\Bigr\rvert &\leq 2 \int_{\mathbb{R}^{2d}} |u(x)|^2|\omega^{-1/2}(k) \chi(k) z(k)| \mathrm{d}x \mathrm{d}k\\
    &\leq 2\|u\|^2_{\mathfrak{H}}\, \|\omega^{-1/2} \chi \|_{\mathfrak{H}} \, \|z\|_{\mathfrak{H}} = 2\delta^2 \|\omega^{-1/2} \chi \|_{\mathfrak{H}} \, \|z\|_{\mathfrak{H}}.
  \end{align*}
  \ref{item:5}: It is convenient to complete the square in the expression of
  the energy functional
  \begin{equation*}
    \begin{split}
      \langle z ,\omega z \rangle_{\mathfrak{H}} + 2\Re\bigl\langle u,  \langle \lambda_x, z  \rangle_{\mathfrak{H}} u \bigr\rangle_{\mathfrak{H}} =\mspace{-5mu}\int_{\mathbb{R}^{2d}}  \mspace{-10mu}|u(x)|^2 \Bigl( \bigl(\overline{z(k)}+ \delta^2 e^{-ik\cdot x} \tfrac{\bar{\chi}(k)}{\omega^{3/2}(k)} \bigr) \tfrac{\omega(k)}{\delta^2} \bigl(z(k)\\+ \delta^2 e^{ik\cdot x} \tfrac{\chi(k)}{\omega^{3/2}(k)} \bigr) \Bigr)\mathrm{d} x \mathrm{d}k  -\delta^4 \|\omega^{-1} \chi\|^2_{\mathfrak{H}}\geq - \delta^4 \|\omega^{-1}\chi\|^2_{\mathfrak{H}}\; .
    \end{split}
  \end{equation*}
  Therefore
  \begin{equation}\label{ineq:laplaceboundfunct}
    \mathscr{E}(u,z)  \geq  \langle u,(-\Delta + V)u \rangle_{\mathfrak{H}} - \delta^4 \|\omega^{-1} \chi\|^2_{\mathfrak{H}}\; .
  \end{equation}
\end{proof}
\begin{proposition}\label{prop:existencemin}
  For any $\delta>0$ there exists $(u_0,z_0) \in \{(u,z)\in \mathscr{X}\,,\,
  \lVert u \rVert_{\mathfrak{H}}^{}=\delta\}$ such that
  \begin{equation}\label{eq:minimizationgse}
    \mathscr{E}(u_0,z_0) = \inf_{\substack{(u,z) \in \mathscr{X}\\\lVert u  \rVert_{\mathfrak{H}}^{}=\delta}} \mathscr{E}(u,z)\; .
  \end{equation}
\end{proposition}
\begin{proof}
  The functional $\mathscr{E}$ is bounded from below, thanks to
  \cref{lem:functionalineq}. Now consider a minimizing sequence $(u_n,z_n)$:
  \begin{equation}
    \mathscr{E}(u_n,z_n) = E_{\delta} +\frac{1}{n} \; .
  \end{equation}
  Since $\|u_n\|_{\mathfrak{H}}^2 = \delta$, there exists a subsequence
  $\{u_{n_k}\}_{k \in \NN} $ and $u_0 \in \mathfrak{H}$ such that $u_{n_k}
  \rightharpoonup u_0$ converges weakly in $\mathfrak{H}$. By
  \eqref{ineq:laplaceboundfunct} it also exists $C>0$ such that
  \begin{equation}
    \langle u_{n_k} ,(-\Delta + V) u_{n_k} \rangle_{\mathfrak{H}} \leq C\; .
  \end{equation}
  By lower semi-continuity of the induced norm, this yields
  \begin{equation}\label{liminf:laplacian}
    \langle u_0 ,(-\Delta + V) u_0\rangle_{\mathfrak{H}} \leq \liminf_{k \rightarrow + \infty}\, \langle u_{n_k} ,(-\Delta + V)u_{n_k} \rangle_{\mathfrak{H}} \leq C\; .
  \end{equation}
  Therefore, $u_0 \in \mathscr{D}(\sqrt{-\Delta + V})$. In addition,
  $\|(-\Delta + V)^{1/2}u_{n_k}\|_{\mathfrak{H}}^2\leq C$, thus there exists
  a subsubsequence $\{u_{n_{k_j}}\}_{j \in \NN}$ and $v \in \mathfrak{H}$
  such that
  \begin{equation}
    (-\Delta+V)^{1/2} u_{n_{k_j}} \rightharpoonup v\; .
  \end{equation}
  Let now $\varphi \in \mathfrak{H}$:
  \begin{equation}
    0 = \lim_{j \rightarrow + \infty } \langle \varphi,u_{n_{k_j}} - u_0 \rangle_{\mathfrak{H}} = \lim_{j \rightarrow + \infty }  \langle (-\Delta + V)^{-1/2} \varphi,(-\Delta + V)^{1/2}(u_{n_{k_j}} - u_0) \rangle_{\mathfrak{H}}\; ,
  \end{equation}
  that implies $v = (-\Delta + V)^{1/2} u_0$. Since the potential $V$ is
  confining and $(-\Delta + V)^{-1/2}$ is compact, one has
  \begin{equation}\label{eq:limitu}
    \underset{j\to \infty}{\mathfrak{H}\text{-lim}}\: u_{n_{k_j}}=u_0\; .
  \end{equation}
  In addition, $\|u_0\|_{\mathfrak{H}} = \delta$.

  Consider now the sequence $\{z_{n_{k_j}}\}_{j \in \NN}$. By
  \cref{lem:functionalineq} together with the fact that for all $k\in
  \mathbb{R}^d$, $\omega(k) \geq m >0$ we have that for all $n\in
  \mathbb{N}$:
  \begin{align*}
    C > \mathscr{E}(u_n,z_n) &\geq \mathscr{E}_0(u_n,z_n) - 2\delta^2 \|\omega^{-1/2}\chi\|_{\mathfrak{H}} \|z_n\|_{\mathfrak{H}} \geq m \|z_n\|^2_{\mathfrak{H}}  - 2\delta^2 \|\omega^{-1/2}\chi\|_{\mathfrak{H}} \|z_n\|_{\mathfrak{H}}\; .
  \end{align*}
  Hence $\|z_{n_{k_j}}\|_{\mathfrak{H}}$ is bounded, so there exist a
  subsubsequence which we still denote by $\{z_{n_{k_{j}}}\}_{j \in \NN}$ and
  a function $z_0 \in \mathfrak{H}$ such that $z_{n_{k_{j}}} \rightharpoonup
  z_0$. As a byproduct, lower semi-continuity for $\langle z , \omega z
  \rangle_{\mathfrak{H}}$ yields
  \begin{equation}\label{liminf:fieldenergy}
    \langle z_0 ,\omega z_0 \rangle_{\mathfrak{H}} \leq \liminf_{j \rightarrow + \infty} \langle z_{n_{k_{j}}} ,\omega z_{n_{k_{j}}} \rangle_{\mathfrak{H}}\; ,
  \end{equation}
  \emph{i.e.}, $z_0 \in \mathscr{D}(\sqrt{\omega})$. In addition, thanks to
  Fubini's theorem and weak convergence,
  \begin{equation}\label{eq:limitrealpart}
    \lim_{j \rightarrow + \infty} \bigl\langle u_0 ,\langle z_{n_{k_{j}}} , \lambda_x \rangle_{\mathfrak{H}} u_0 \bigr\rangle_{\mathfrak{H}} = \bigl\langle u_0 ,\langle z_0 , \lambda_x \rangle_{\mathfrak{H}} u_0 \bigr\rangle_{\mathfrak{H}}\,.
  \end{equation}
  Furthermore,
  \begin{equation*}
    \begin{split}
      \Bigl\lvert\bigl\langle u_{n_{k_{j}}} , \Re\langle z_{n_{k_{j}}}, \lambda_{x} \rangle_{\mathfrak{H}} u_{n_{k_{j}}} \bigr\rangle_{\mathfrak{H}}  - \bigl\langle u_0 ,\Re\langle z_0,\lambda_{x} \rangle_{\mathfrak{H}} u_0 \bigr\rangle_{\mathfrak{H}}\Bigr\rvert  \leq C \bigl\|\lambda_{(\cdot)}\bigr\|_{L^{\infty}(\RR^d; \mathfrak{H})} \bigl\lVert u_{n_{k_{j}}}-u_0  \bigr\rVert_{\mathfrak{H}}^2  \\+ \bigl\|\langle z_{n_{k_{j}}} - z_0,\lambda_{(\cdot)} \rangle_{\mathfrak{H}}\bigr\|_{L^{\infty}_x} \bigl\|u_{0}\bigr\|^2_{\mathfrak{H}}\; ,
    \end{split}
  \end{equation*}
  where the right hand side converges to zero thanks to the strong
  convergence of $u_{n_{k_j}}$ and \eqref{eq:limitrealpart}. Thus,
  \begin{equation}\label{liminf:interaction}
    \lim_{j \rightarrow + \infty} \bigl\langle u_{n_{k_{j}}} ,  \Re\langle z_{n_{k_{j}}}, \lambda_{x} \rangle_{\mathfrak{H}} u_{n_{k_{j}}} \bigr\rangle_{\mathfrak{H}} = \bigl\langle u_0 , \Re \langle z_0, \lambda_{x} \rangle_{\mathfrak{H}} u_0 \bigr\rangle_{\mathfrak{H}}\; .
  \end{equation}
  Hence, we conclude the proof:
  \begin{equation}
    E_{\delta}\leq \mathscr{E}(u_0,z_0) \leq \liminf_{j\rightarrow + \infty} \mathscr{E}(u_{n_{k_{j}}},z_{n_{k_{j}}})=E_{\delta}+\liminf_{j \rightarrow + \infty}\tfrac{1}{n_{k_{j}}} = E_{\delta}\; .
  \end{equation}
\end{proof}

The next natural question is whether the minimizer $(u_0,z_0)$ is unique,
apart from the trivial $U(1)$ symmetry on the nucleon part. That amounts to
ask whether the ground state is invariant with respect to the eventual
symmetries of the external potential $V$. In order to investigate such
question it is convenient to recast our minimization problem in the variables
$(u,z)$ in an equivalent way involving only the variable $u$. In order to do
that, let us state (without proof) a preparatory lemma.

\begin{lemma}
  The energy functional $\mathscr{E}: \mathscr{X} \longrightarrow \RR$ is
  continuous and G\^ateaux differentiable.
\end{lemma}

The derivative of $\mathscr{E}$ yields:
\begin{equation}
  \label{eq:gateauxderenergy}
  \begin{split}
    \frac{\der}{\der \alpha}\bigg|_{\alpha =0} \mathscr{E}(u + \alpha v, z + \alpha y) =  2 \Re \langle u ,(-\Delta + V)v \rangle_{\mathfrak{H}} + 2 \Re \langle z, \omega y \rangle_{\mathfrak{H}} +2\Re\bigl\langle u , \langle y, \lambda_x \rangle_{\mathfrak{H}} u \bigr\rangle_{\mathfrak{H}} \\+ 2\Re \bigl\langle v, 2\Re \langle \lambda_x, z \rangle_{\mathfrak{H}} u \bigr\rangle_{\mathfrak{H}}\; .
  \end{split}
\end{equation}
The above is true for any variation $(v,y) \in \mathscr{X}$, and in
particular for $v=0$:
\begin{equation}\label{eq:criticalpoints}
  \frac{\der}{\der \alpha}\bigg|_{\alpha =0} \mathscr{E}(u, z + \alpha y) =  2 \Re \langle z, \omega y \rangle_{\mathfrak{H}} + 2\Re\bigl\langle u , \langle y, \lambda_x \rangle_{\mathfrak{H}}u \bigr\rangle_{\mathfrak{H}}\; .
\end{equation}

\begin{proposition}\label{prop:equivminproblems}
  If $(u_0,z_0) \in \mathscr{X}$ is a minimizer of the energy functional
  $\mathscr{E}(u,z)$ under the constraint $\lVert u_0
  \rVert_{\mathfrak{H}}^{}=\delta$, then
  \begin{equation}\label{eq:z_0expression}
    z_0(\,\cdot\,) = -\frac{\chi}{\omega^{3/2}}(-\,\cdot\,)\;  \overline{\widehat{|u_0|^2} }(\,\cdot\,)\; .
  \end{equation}
  Furthermore,
  \begin{equation}
    \mathscr{E}(u_0,z_0)=\inf_{\substack{(u,z)\in \mathscr{X}\\\lVert u  \rVert_{\mathfrak{H}}^{}=\delta}} \mathscr{E}(u,z) = \inf_{\substack{u \in \mathscr{D}(\sqrt{-\Delta  +V}) \\ \|u\|_{\mathfrak{H}}= \delta}} \mathcal{E}(u)=\mathcal{E}(u_0)\; ,
  \end{equation}
  where $\mathcal{E}$ is the Hartree functional
  \begin{equation}
    \label{eq:hartree}
    \mathcal{E}(u) = \langle u, (-\Delta + V)  u \rangle_{\mathfrak{H}} - \int_{\mathbb{R}^{2d}}^{} \lvert u(x)  \rvert_{}^2 \,W(x-y)\,\lvert u(y)  \rvert_{}^2 \mathrm{d}x \mathrm{d}y
  \end{equation}
  with pair potential $W\in \mathscr{S}(\mathbb{R}^d,\mathbb{R})$ satisfying
  \begin{equation*}
    W=\widehat{\bigg( \frac{\bar{\chi}(\,\cdot\,)\chi(-\,\cdot\,)}{\omega^2}\bigg)}\; .
  \end{equation*}
\end{proposition}
\begin{proof}
  According to \eqref{eq:criticalpoints}, $(u_0,z_0)\in\mathscr{X}$ are
  critical points of the functional $\mathscr{E}(u,z)$ and have to satisfy,
  for all $y\in \mathscr{D}(\sqrt{\omega})$,
  \begin{equation}
    \Re \bigl\langle y , \omega z_0 + \langle u_0, \lambda_x u_0 \rangle_{\mathfrak{H}_x} \bigr\rangle_{\mathfrak{H}} = 0\; .
  \end{equation}
  This implies
  \begin{align*}
    z_0 (k) = -\omega^{-1}(k) \langle u_0, \lambda_x (k) u_0 \rangle_{\mathfrak{H}_x} =  -\frac{\chi(-k)}{\omega^{3/2}(k)} \widehat{|u_0|^2} (-k)\in\mathscr{D}(\sqrt{\omega})\; ,
  \end{align*}
  and proves \eqref{eq:z_0expression}. Computing $\mathscr{E}(u,z)$ with the
  constraint
  \begin{equation}
    \label{eq:const-eq}
    z(\,\cdot\,) = -\frac{\chi}{\omega^{3/2}}(-\,\cdot\,) \;\widehat{|u|^2} (-\,\cdot\,)\; ,
  \end{equation}
  yields
  \begin{align*}
    \mathscr{E}(u,z)=\mathcal{E}(u)\;.
  \end{align*}
  Hence, it follows that
  \begin{align*}
    \mathscr{E}(u_0,z_0)= \inf_{\substack{(u,z)\in \mathscr{X}\\\lVert u  \rVert_{\mathfrak{H}}^{}=\delta}} \mathscr{E}(u,z) \leq \inf_{\substack{u\in \mathscr{D}(\sqrt{-\Delta  +V})\\ \eqref{eq:const-eq}\,, \,\lVert u  \rVert_{\mathfrak{H}}^{}=\delta}} \mathscr{E}(u,z)=
    \inf_{\substack{u \in \mathscr{D}(\sqrt{-\Delta  +V})\\ \|u\|_{\mathfrak{H}}= \delta}} \mathcal{E}(u)\leq  \mathcal{E}(u_0)\; .
  \end{align*}
\end{proof}

\begin{lemma}
  \label{lem.lagmul}
  If $u_0\in \mathscr{D}(\sqrt{-\Delta +V})$ is a minimizer of the Hartree energy
  \eqref{eq:hartree},
  \begin{equation}
    E_\delta=  \inf_{\substack{u \in \mathscr{D}(\sqrt{-\Delta  +V}) \\ \|u\|_{\mathfrak{H}}= \delta}} \mathcal{E}(u)\; ,
  \end{equation}
  then $u_0\in \mathscr{D}(-\Delta +V)$ and there exists $\lambda=\delta^{-2} E_\delta$ such that
  \begin{equation}
    (-\Delta  +V ) u_0- W*|u_0|^2 u_0= \lambda u_0\,.
  \end{equation}
\end{lemma}
\begin{proof}
  Follows by the generalized Lagrange multiplier theorem on Banach spaces,
  see \citep[see \emph{e.g.}][Theorem 2 \S 8.11]{Lusternik1974}.
\end{proof}

The uniqueness of ground states for the Hartree functional $\mathcal{E}$ has
been thoroughly investigated, and it is related to symmetry breaking
phenomena \citep[see][and references therein
contained]{aschbacher2002jmp,correggi2007jmp,correggi2012jmp,seiringer2011cmp}.
\begin{lemma}
  \label{lemma:1}
  There exists $\delta^{*}>0$ such that for any $0<\delta<\delta^{*}$, the Hartree
  functional $\mathcal{E}(\cdot)$ given in \eqref{eq:hartree} with the constraint
  $\lVert \cdot \rVert_{\mathfrak{H}}^{}=\delta$ admits a unique minimizer up to $U(1)$
  symmetry.
\end{lemma}
\begin{proof}
  The proof follows for instance by \citep[Theorem 2]{aschbacher2002jmp}
  although the assumptions there are slightly different. So for completeness,
  we sketch the main argument.  By rescaling, the original variational
  problem is put in a perturbative form,
$$
\inf_{\substack{u \in \mathscr{D}(\sqrt{-\Delta +V}) \\ \|u\|_{\mathfrak{H}}= \delta}}
\mathcal{E}(u)= \delta^2 \inf_{\substack{u \in \mathscr{D}(\sqrt{-\Delta +V}) \\
    \|u\|_{\mathfrak{H}}= 1}} \mathcal{E}_\delta(u)\,,
$$ 
with \begin{eqnarray*} \mathcal{E}_\delta(u)&=& \langle u, (-\Delta + V) u
  \rangle_{\mathfrak{H}} - \delta^2\int_{\mathbb{R}^{2d}}^{} \lvert u(x)
  \rvert_{}^2 \,W(x-y)\,\lvert u(y) \rvert_{}^2 \mathrm{d}x \mathrm{d}y\,.
\end{eqnarray*}
Define, for $u\in \mathfrak{H}$, the self-adjoint operator
$$
H_\delta^{(u)}= -\Delta + V- \delta^2 \,\bigl(W*|u|^2\bigr)\,.
$$
Thanks to the hypothesis \eqref{hyp:pot}, the ground state energy of the
Schr\"odinger operator $H_0^{(u)}=-\Delta+V$ is non-degenerate (see for
instance \cite[Theorem 1]{MR381571}). Moreover, for all $u\in \mathfrak{H},
\|u\|_\mathfrak{H}\leq 1$, $\delta\mapsto H_\delta^{(u)}$ is an analytic
family of type A with a potential satisfying
$$
\|W*|u|^2 \|_{L^\infty}\leq \|W\|_{L^\infty} \,.
$$
Hence, according to regular perturbation theory there exists $\delta^*>0$
sufficiently small and independent from $u$ such that the ground state energy
of $H_\delta^{(u)}$ stays non-degenerate for all $0<\delta<\delta^*$.  So,
let us denote by $P^{(u)}_\delta$ the one dimensional orthogonal projector on
the ground state of $H_\delta^{(u)}$.  Then using \cref{lem.lagmul}, one
remarks that any minimizer of $\mathcal{E}_\delta(\cdot)$ is up to a phase
factor a fixed point of the following map
\begin{eqnarray*}
  T: \mathcal{S}=\{u\in  \mathfrak{H}, \|u\|_\mathfrak{H}=1\} & \longrightarrow & \mathcal{S} \\
  u &\longrightarrow  & \frac{P_\delta^{(u)}\psi_0}{\|P_\delta^{(u)}\psi_0\|_\mathfrak{H}}\,,
\end{eqnarray*}
where $\psi_0$ is a ground state of $H_0^{(u)}$. Again thanks to perturbation
theory, one proves that $T$ is a well defined (strict) contraction admitting
therefore a unique fixed point. Combining this to the existence of minimizers
given in \cref{prop:equivminproblems}, proves the claimed result.

\end{proof}

% \begin{remark} Symmetry breaking, and thus possible non-uniqueness of the
%   minimizer is expected for $\delta>\delta^{*}$, and it has been proved for
%   suitable non-trapping potentials $V\in
%   \mathscr{C}_{\mathrm{b}}(\mathbb{R}^d)$.
% \end{remark}

Using \cref{lemma:1}, we can further characterize the Wigner measures of
quantum ground states when the classical minimizer is unique.
\begin{corollary}[of \cref{thm:2}]
  \label{cor:5}
  Consider $\big\{\Phi^{(n_k)}_{\hslash_k}\big\}_{k\in \N}$ to be a sequence of ground
  states of $H_{\hslash_k}^{(n_k)}$ such that there exists $0<\delta<\delta^{*}$ (with
  $\delta^{*}$ given in \cref{lemma:1}) such that:
  \begin{equation}
    \label{eq.gdshyp}
    N_1 \Phi^{(n_k)}_{\hslash_k}= \hslash_k \;n_k\; \Phi^{(n_k)}_{\hslash_k}\,,\qquad \lim_{k\to\infty} \hslash_k=0\,,\qquad \text{ and } \qquad  \lim_{k\to\infty} \hslash_kn_k=\delta^2\,.
  \end{equation}
  In addition, let $u_0$ be the unique minimizer modulo $U(1)$-invariance of
  $\mathcal{E}(u)$ with $\lVert u_0 \rVert_{\mathfrak{H}}^{}=\delta$. Then the corresponding
  Wigner measure is unique, and explicit:
  \begin{equation*}
    \begin{split}
      \mathscr{M}\bigl(\Phi^{(n_k)}_{\hslash_k}; k\in\N\bigr)=   \{\mu\} =   \Bigl\{  \frac{1}{2\pi} \int_0^{2\pi} \;  \delta_{e^{i\theta}u_0}\otimes \delta_{  -\frac{\chi(-\,\cdot\,)}{\omega^{3/2}} \overline{\widehat{|u_0|^2}}(\,\cdot\,)}
      \; {\rm d}\theta \Bigr\}\; .
    \end{split}
  \end{equation*}
\end{corollary}
\begin{proof}
  According to \cref{thm:2} and Lemma \ref{lemma:1}, any Wigner measure $\mu$
  of a sequence of ground states $\big\{\Phi^{(n_k)}_{\hslash_k}\big\}_{k\in
    \N}$ satisfying \eqref{eq.gdshyp}, concentrates on the set of minimizers
$$
\mathscr{X}^0_\delta:=\Bigl\{\big(e^{i\theta} u_0, -\frac{\chi(-\,\cdot\,)}{\omega^{3/2}}
\overline{\widehat{|u_0|^2}}(\,\cdot\,)\big), \theta\in[0,2\pi[\Bigr\}\,.
$$
On the other hand, the probability measure $\mu$ is $U(1)$ invariant,
\emph{i.e.}, for all $\theta\in[0,2\pi[$
\begin{equation}
  \label{eq.invmu}
  (e^{i\theta})\, _{*}\,\mu=\mu\,.
\end{equation}
Such property is indeed a consequence of \eqref{wigner} and the following
relation which holds true at least for a subsequence:
\begin{eqnarray*}
  \widehat{(e^{i\theta})\, _{*}\,\mu} (\eta_1\oplus\eta_2)  &=& \lim_{k\to\infty} \langle e^{i\theta N_1} \Phi_{\hslash_k}^{(n_{k})} , W(\eta_1\oplus\eta_2) e^{i\theta N_1}\Phi_{\hslash_k}^{(n_{k})} \rangle_{}\\
  &=&\lim_{k\to\infty} \langle  \Phi_{\hslash_k}^{(n_{k})} , W(\eta_1\oplus\eta_2)\Phi_{\hslash_k}^{(n_{k})} \rangle_{}\\ &=& \widehat\mu(\eta_1\oplus\eta_2)\,\,,
\end{eqnarray*}
where here $\widehat{(e^{i\theta})\, _{*}\,\mu}$ and $\widehat\mu$ denote
respectively the generating functions of $(e^{i\theta})\, _{*}\,\mu$ and
$\mu$.  So, this implies that $\mu$ is translation-invariant with respect to
the variable $\theta$ and hence $\mu$ is uniformly distributed over the
circle $U(1)$. Therefore, one concludes
$$
\mu=\frac{1}{2\pi} \int_0^{2\pi} \; \delta_{e^{i\theta}u_0}\otimes \delta_{
  -\frac{\chi(-\,\cdot\,)}{\omega^{3/2}} \overline{\widehat{|u_0|^2}}(\,\cdot\,)} \; {\rm
  d}\theta \,.
$$

\end{proof}

\section{Some open problems}
\label{sec.op}
We end up this paper by discussing some open problems, in the hope they will
stimulate a renewed interest on such fundamental and mathematically
challenging topic.

\begin{enumerate}
\item \emph{Asymptotic completeness for the S-KG system:} The scattering
  theory for the S-KG equation \eqref{eq:skg} discussed in Section
  \ref{sec:3} is far from complete. In particular, it will be interesting to
  characterize the out/in classical asymptotic radiationless solutions in
  $\mathscr{K}^{\pm}_0$ and to prove the ``weak'' asymptotic completeness
 $$
 \mathscr{K}^{+}_0=\mathscr{K}^{-}_0\,,
  $$
  or the ``strong'' asymptotic completeness
  $$
  \mathscr{K}^{+}_0=\{\text{nonlinear bound states or
    solitons}\}=\mathscr{K}^{-}_0\,.
  $$
  Actually, in light of J.\ Dereziński and C.\ Gérard result
  \cite{derezinski1999rmp} for quantum asymptotic completeness, such
  statements are expected to hold true.

\item \emph{Radiative decay:} Radiative decay is the fundamental physical
  mechanism that explains relaxation of excited atoms to their radiationless
  ground states. According to the work of M.\ Hubnert and H.\ Spohn
  \cite{MR1326139}, it can be stated roughly as the convergence
  \begin{equation}
    \label{eq.op.1}
    \lim_{t\to  \infty} \langle e^{it H_\hslash^{(n)}} \Psi_\hslash^{(n)} ,\, A \, e^{it H_\hslash^{(n)}} \Psi_\hslash^{(n)}\rangle = \langle \Phi_\hslash^{(n)}, \, A \, \Phi_\hslash^{(n)}\rangle\,,
  \end{equation}
  where $A$ and $\Psi_\hslash^{(n)}$ are respectively some given observables and
  states, and $\Phi_\hslash^{(n)}$ is the ground state. For the above identity to be
  true, one assumes further that the pure point spectrum space
  $\mathscr{H}_{pp}^{(n)}=\mathds{1}_{\mathrm{pp}}(H_{\hslash}^{(n)}) \mathscr{H}^{(n)}$ is
  one dimensional. Thus, a formal semiclassical limit in \eqref{eq.op.1}
  suggests the identity
  \begin{equation}
    \label{eq.op.rad}
    \lim_{t\to\infty} \int_{\mathscr Z} a(z) \, d\mu_t(u,z)=\int_{\mathscr Z} a(z) \, d\mu_0(u,z)\,,
  \end{equation}
  where $\mu_0$ is a measure concentrated on the classical ground state of the
  S-KG equation, given for instance by \cref{cor:5}.  Proving rigorously
  \eqref{eq.op.rad} will highlight a relaxation phenomena at the level of the
  nonlinear classical system.

\item \emph{Ground states convergence:} Convergence of the ground state
  energy of the Yukawa model towards the infimum of the classical energy of
  the S-KG system is discussed in \cref{sec:semicl-limit-ground}. It is
  interesting to extend such result with the concentration property in
  \cref{thm:2} to various models of particle-field interactions, which could
  carry the following major difficulties:
  \begin{itemize}
  \item [a)] \emph{Infrared problems:} Ground states of models with massless
    fields are well studied \citep[see,
    \emph{e.g.},][]{MR1623746,MR2077252,MR2103217}.  The massless assumption
    amounts to consider $m=0$ in \eqref{hyp:omega} and it induces some lack
    of compactness with respect to the field variables. In particular, the
    proof of convergence in \cite{ammari2014jsp} fails, and a different
    argument needs to be worked out.
  \item [b)] \emph{Ultraviolet problems:} The removal of the ultraviolet
    cutoff (\emph{i.e.}, taking $\chi\equiv 1$ in assumption \eqref{hyp:lambda})
    requires a renormalization procedure, due to the nucleons' self-energy
    divergence. Here as well the semiclassical convergence of the ground
    state energy needs to be worked out with different arguments \citep[see,
    \emph{e.g.},][]{MR3912797}.
  \item [c)] \emph{Translation invariance problem:} In the absence of a
    confining potential $V$, the Yukawa model becomes translation-invariant
    and the quantum Hamiltonian as well as the classical energy can be
    fibered with respect to its total momentum. Therefore, it is natural to
    ask whether the ground state energy and ground states at fixed momentum
    converge, in the semiclassical limit, towards their classical
    counterparts.
  \end{itemize}

\item \emph{Transition and scattering amplitudes:} Is it possible to derive
  semiclassical limits for transition and scattering amplitudes of various
  models (carrying in particular the infrared, ultraviolet or
  translation-invariance difficulties discussed in the previous point)? It is
  worth noting that the quantum scattering theory is studied for instance in
  \cite{MR700181,MR1878286,MR2077252,MR3201223} for massless fields, in
  \cite{MR2764336,MR3411743} for the translation invariant Nelson model, and
  in \cite{MR1809881} for the renormalized Nelson model without ultraviolet
  cutoff.  Another model where such questions could be addressed, perhaps in
  a simpler way, is the so-called (massive or massless) spin-boson model
  \citep[see, \emph{e.g.},][]{MR4087255}.

\item \emph{Time decay and resolvent estimates:} By drawing further the
  parallel between semi-classical analysis of infinite and finite dimensional
  systems, we can speculate about the characterization of uniform time decay
  and resolvent estimates for the quantum Yukawa theory. In particular, it is
  quite interesting to prove the following uniform resolvent estimate for all
  $\hslash\in(0,1]$:
  \begin{equation}
    \label{eq.op.3}
    \|\langle A\rangle^{-s} (H_\hslash^{(n)}-\lambda\pm i0)^{-1}  \langle A\rangle^{-s}  \|_{\mathscr{B}(\mathscr{H}^{(n)})}  \leq C_1(\hslash) \,,
  \end{equation}
  where $A$ is a given selfadjoint operator and the $\hslash$-dependence
  $C_1(\hslash)$ is to be quantitatively determined. Moreover, it is also
  interesting to derive the following uniform decay estimate for all
  $t\in\mathbb R$ and $\hslash\in(0,1]$:
  \begin{equation}
    \label{eq.op.2}
    \|\langle A\rangle^{-s}\,\chi(H_\hslash^{(n)}) \, e^{it\hslash^{-1} H_\hslash^{(n)}}  \langle A\rangle^{-s}  \|_{\mathscr{B}(\mathscr{H}^{(n)})}  \leq C_2(\hslash) \,\langle t\rangle^{-\varepsilon}\,,
  \end{equation}
  where $\chi\in\mathscr{C}_0^\infty(\mathbb R)$. The exponents $s,\varepsilon>0$ as well as the
  $\hslash$-dependence $C_2(\hslash)$ need to be determined as well.  Such estimates
  \eqref{eq.op.3}-\eqref{eq.op.2} are well known for the Schr\"odinger
  operator \citep[see, \emph{e.g.},][and references therein]{MR927007,
    MR1221361}, and they are related to Mourre's theory, where $A$ represents
  a conjugate or a position operator. Note that
  \eqref{eq.op.3}-\eqref{eq.op.2} are expected to hold when the cutoff
  function $\chi$, respectively $\lambda$, are localized on radiating (``non
  trapping'') energy levels of the S-KG equation.  Indeed, the latter
  inequalities \eqref{eq.op.3}-\eqref{eq.op.2} can in principle be derived
  from a semiclassical positive commutator (Mourre) estimate
  $$
  \chi(H_\hslash^{(n)}) \,[H,iA]\, \chi(H_\hslash^{(n)})\geq C(\hslash)
  \,\chi^2(H_\hslash^{(n)}) \,, \quad \forall \hslash\in (0,1]\,,
  $$
  and a semiclassical limiting absorption principle. Note that a Mourre
  estimate is proved in \cite{derezinski1999rmp} for the ultraviolet-cutoff
  Nelson model with $\hslash=1$, and a fixed number of nucleons $n$.
\end{enumerate}

{\footnotesize \newcommand{\etalchar}[1]{$^{#1}$}

\end{document}